\documentclass[letterpaper,twocolumn,10pt]{article}
\usepackage{usenix}

\usepackage{amsmath}
\usepackage{amsthm}
\usepackage{graphicx}		
\usepackage{array}
\usepackage{xfrac}  

\usepackage{thm-restate}
\usepackage{dblfloatfix}
\usepackage{etoolbox}

		\hypersetup{
			colorlinks=true,
			linkcolor=[rgb]{0,0,.5},
			urlcolor=[rgb]{0,0,.5},
			citecolor=[rgb]{0,0,.5},
			pdfstartview=FitH}

\usepackage[hypcap=true]{subcaption}

\usepackage{todonotes}
\usepackage{xurl}
\usepackage{enumerate}
\usepackage{amsfonts,bm,bbm,paralist,nicefrac,comment,float,thmtools,thm-restate,comment,multirow}
\usepackage[export]{adjustbox}
\usepackage{xspace}
\usepackage[inline]{enumitem}
\usepackage{algorithm}
\usepackage[noend]{algpseudocode}
\usepackage{eqparbox}
\usepackage{xparse}
\usepackage{cleveref}
\usepackage{booktabs}

\usepackage{tikz}
\usetikzlibrary{arrows.meta, positioning}
\usetikzlibrary{shapes.geometric}
\declaretheorem[name=Lemma]{lem}

\usepackage{colortbl}
\usepackage[table]{xcolor}

\usepackage{listings}
\newif\ifonlineversion
\onlineversiontrue

\crefname{section}{Section}{Sections}
	\crefformat{section}{#2Section~#1#3}
\crefname{subsection}{Section}{Sections}
	\crefformat{subsection}{#2Section~#1#3}
\crefname{appendix}{Appendix}{Appendices}
	\crefformat{appendix}{#2Appendix~#1#3}
\crefname{equation}{}{}

\crefname{figure}{Figure}{Figures}
	\crefformat{figure}{#2Figure~#1#3}
\crefname{theorem}{Theorem}{Theorems}
	\crefformat{theorem}{#2Theorem~#1#3}
\crefname{assumption}{Assumption}{Assumptions}
	\crefformat{assumption}{#2Assumption~#1#3}
\crefname{lemma}{Lemma}{Lemmas}
	\crefformat{lemma}{#2Lemma~#1#3}
\crefname{proposition}{Proposition}{Propositions}
	\crefformat{proposition}{#2Proposition~#1#3}
\crefname{corollary}{Corollary}{Corollaries}
	\crefformat{corollary}{#2Corollary~#1#3}
\crefname{observation}{Observation}{Observations}
	\crefformat{observation}{#2observation~#1#3}
\crefname{remark}{Remark}{Remarks}
	\crefformat{remark}{#2Remark~#1#3}
\crefname{example}{Example}{Examples}
	\crefformat{example}{#2Example~#1#3}
\crefname{table}{Table}{Tables}
	\crefformat{table}{#2Table~#1#3}
\crefname{definition}{Definition}{Definitions}
	\crefformat{definition}{#2Definition~#1#3}

\newcommand{\createappendix}[1]{\expandafter\newcommand\csname #1\endcsname{}}

\ifonlineversion
\newcommand{\addtoappendix}[4]{#3}
\else
\newcommand{\addtoappendix}[4]{\expandafter\appto\csname #1\endcsname{#2#3} #4}
\fi

\ifonlineversion
\newcommand{\printappendix}[1]{}
\else
\newcommand{\printappendix}[1]{\csname #1\endcsname}
\fi

\newcommand{\printsymbol}[1]{\csname #1\endcsname{}{}{}{} & \texttt{#1}}

\algrenewcommand{\algorithmiccomment}[1]{\hfill $\triangleright$ \eqparbox{COMMENT\thealgorithm}{\footnotesize{#1}}}
\algnewcommand{\LComment}[1]{\hfill $\triangleright$ \begin{minipage}[t]{\eqboxwidth{COMMENT\thealgorithm}}\footnotesize{#1}\strut\end{minipage}}
\algnewcommand{\algrule}[1][.2pt]{\par\vskip.5\baselineskip\hrule height #1\par\vskip.5\baselineskip}
\algnewcommand{\algdescription}[1]{\Statex \footnotesize{\textbf{Description: }#1}}
\algnewcommand{\algincentives}[1]{\Statex \footnotesize{\textbf{Incentive Structure: }#1}}

\renewcommand{\paragraph}[1]{\medskip\noindent\textbf{#1}}
\newcommand{\eat}[1]{}

\newtheorem{theorem}{Theorem}
\newtheorem{lemma}[lem]{Lemma}

\newtheorem{proposition}{Proposition}
\newtheorem{remark}{Remark}
\newtheorem{example}{Example}
\newtheorem{definition}{Definition}

\microtypecontext{spacing=nonfrench}
\begin{document}
\date{}
\title{\Large \bf Perils of Parallelism: Transaction Fee Mechanisms \\ under Execution Uncertainty}
\author{
{\rm Sarisht Wadhwa\thanks{Equal contribution.}}\\
Duke University
\and
{\rm Aviv Yaish\footnotemark[1]}\\
Yale University, IC3, \\Complexity Science Hub Vienna
\and
{\rm Fan Zhang}\\
Yale University, IC3
\and
{\rm Kartik Nayak}\\
Duke University
}

\maketitle

\begin{abstract}
Modern blockchains increasingly rely on parallel execution to improve throughput.
We show several industry and academic transaction fee mechanisms (TFMs) struggle to simultaneously account for execution parallelism while remaining performant and fair.
First, if parallelism affects fees, adversarial protocol manipulations that offset possible benefits to throughput by introducing \emph{fake} transactions become rational: users can insert functionally useless parallel transactions solely to reduce fees, and schedulers can create useless sequential transactions to increase revenue.
Execution contingency, a core feature of expressive programming languages, both exacerbates the aforementioned threats and introduces new ones:
(1) users may overpay for unused resources,
and (2) scheduler revenue is harmed when reserved scheduling slots go unused due to contingency.
We introduce a framework for this challenging setting, and prove an impossibility, highlighting an inherent tension: both parallelism and contingency involve a trade-off between minimizing risks for users and schedulers, as favoring one comes at the expense of the other.
To complete the picture, we introduce fee mechanisms and prove that they achieve the boundaries of this trade-off.
Our results provide rigorous foundations for evaluating designs advanced by notable blockchains, such as Sui and Monad.
\end{abstract}

\newcommand{\txswap}{\ensuremath{tx_{\mathsf{A}}}\xspace}
\newcommand{\opx}{\ensuremath{o_{\mathsf{px}}}\xspace}
\newcommand{\opool}{\ensuremath{o_{\mathsf{pool}}}\xspace}
\newcommand{\obal}{\ensuremath{o_{\mathsf{bal}}}\xspace}

\newcommand{\tx}{\ensuremath{{tx}}\xspace}
\newcommand{\txfee}{\ensuremath{{f}}\xspace}
\newcommand{\gasfee}{\ensuremath{g}\xspace}
\newcommand{\attainablefee}{\ensuremath{f_\mathrm{att}}\xspace}
\newcommand{\maxrecfee}{\ensuremath{r_\mathrm{max}}\xspace}
\newcommand{\minrecfee}{\ensuremath{r_\mathrm{min}}\xspace}
\newcommand{\basefee}{\ensuremath{f_\mathrm{base}}\xspace}
\newcommand{\uifee}{\ensuremath{f_\mathrm{ui}}\xspace}
\newcommand{\paidfee}{\ensuremath{f_\mathrm{act}}\xspace}
\newcommand{\receivedfee}{\ensuremath{r_\mathrm{act}}\xspace}
\newcommand{\burnratio}{\ensuremath{\gamma}\xspace}
\newcommand{\ins}{\text{ins}}
\newcommand{\object}{\ensuremath{o}\xspace}
\newcommand{\setobjects}{\ensuremath{\mathcal{O}}\xspace}
\newcommand{\numobjects}{\ensuremath{n_O}\xspace}
\newcommand{\writeset}{{\ensuremath{W}}\xspace}
\newcommand{\readset}{{\ensuremath{R}}\xspace}
\newcommand{\exec}{\ensuremath{\mathrm{Exec}}\xspace}
\newcommand{\outp}{\ensuremath{\mathrm{out}}\xspace}
\newcommand{\val}[1]{\ensuremath{\mathrm{val}(#1)}\xspace}
\newcommand{\conflicts}[2]{\ensuremath{C(#1, #2)}\xspace}
\newcommand{\gas}[1]{\ensuremath{\textsf{gas}_{#1}}}
\newcommand{\gaslimit}{\ensuremath{G}\xspace}
\newcommand{\state}{\ensuremath{s}\xspace}
\newcommand{\contingentread}{\ensuremath{{C_\readset}}\xspace}
\newcommand{\contingentreadobject}{\ensuremath{{c_\readset}}\xspace}
\newcommand{\contingentwrite}{\ensuremath{{C_\writeset}}\xspace}
\newcommand{\contingentwriteobject}{\ensuremath{{c_\writeset}}\xspace}
\newcommand{\used}{\ensuremath{X}\xspace}
\newcommand{\usedread}{\ensuremath{X_\readset}\xspace}
\newcommand{\usedwrite}{\ensuremath{X_\writeset}\xspace}
\newcommand{\schedule}{\ensuremath{\mathcal{S}}\xspace}
\newcommand{\consensusset}{\ensuremath{\mathcal{T}}\xspace}
\newcommand{\notused}{\ensuremath{\gamma}\xspace}
\newcommand{\makespan}[1]{\ensuremath{v(#1)}\xspace}
\newcommand{\greedy}{\ensuremath{\textsf{GREEDY}}\xspace}
\newcommand{\opt}{\ensuremath{\textsf{OPT}}\xspace}
\newcommand{\random}{\ensuremath{\textsf{RANDOM}}\xspace}
\newcommand{\compunits}{\ensuremath{t}\xspace}
\newcommand{\priorityschedules}{\ensuremath{P}\xspace}
\newcommand{\ratio}{\ensuremath{\alpha}\xspace}
\newcommand{\txcomplimit}{\ensuremath{T_{\textsf{max}}}\xspace}
\newcommand{\numcores}{\ensuremath{n}\xspace}
\newcommand{\accept}{\ensuremath{\textsf{ACCEPT}}\xspace}
\newcommand{\reject}{\ensuremath{\textsf{REJECT}}\xspace}
\newcommand{\userrisk}{\ensuremath{\mathrm{UR}}\xspace}
\newcommand{\schedrisk}{\ensuremath{\mathrm{SR}}\xspace}
\newcommand{\revenue}[1]{\ensuremath{R_{#1}}\xspace}
\newcommand{\classP}{\ensuremath{{\textbf{\emph{P}}}}\xspace}
\newcommand{\classNC}{\ensuremath{{\textbf{\emph{NC}}}}\xspace}
\newcommand{\schedutility}{\ensuremath{U_{\text{sched}}}\xspace}
\newcommand{\block}{\ensuremath{b}\xspace}
\newcommand{\usage}{\ensuremath{\mathcal{U}}\xspace}

\section{Introduction}
\label{sec:introduction}

As the throughput of consensus protocols continues to improve, the throughput bottleneck of blockchains has switched to execution \cite{ponnapalli2021rainblock}.
Parallel execution, i.e., running multiple transactions on parallel threads, is a prominent approach to increase execution throughput either used currently or slated for future upgrades by chains like Sui~\cite{suigas}, Solana~\cite{solanafoundation_2019}, Monad~\cite{Monad}, and Ethereum~\cite{AccessLists}.
However, in blockchains supporting Turing-complete smart contracts, resource contention among transactions cannot be decided without pre-execution.
As a result, contention may undermine parallelism's benefits.
E.g., if two \textit{interdependent} transactions are scheduled to execute in parallel, one of the two would either have to fail and re-execute later (as in Monad~\cite{Monad}), or stall and lock system objects until it can resume executing.
Transaction failure harms the system's user experience, while locking introduces overhead and forces serialization at contention points.

In both cases, execution is either slowed down or system's capacity is not fully used, defeating the purpose of parallelism.
Worse, current designs do not charge transactions for locking capacity that eventually goes unused, thereby harming revenue \cite{yaish2024speculative}.
A proper fee mechanism can alleviate these issues by pricing the externalities of contention and charging transactions that contribute significantly to the critical path.
Thus, a natural question arises: \emph{can we design a fee mechanism that accounts for the impact of transactions on parallel execution?}

To answer this, we must consider two ``perils'' that make the question hard: contingent transactions and shill transactions.

\paragraph{Peril 1: Contingent Transactions and Risk.}
Parallel execution amplifies the impact of \emph{contingent transactions}: transactions which can access different object lists based on the state when executing.
Such transactions create a tension between two forms of risk:
\emph{user risk}, where users overpay for objects that go unused, and
\emph{scheduler risk}, where execution capacity which is reserved ex-ante goes unpaid due to under-execution.
We formalize both notions and prove a central impossibility result: 
\emph{no mechanism can simultaneously eliminate both user risk and scheduler risk without either executing transactions in advance or collapsing to constant fees.}

\paragraph{Risk Allocation Mechanisms.}
Given this impossibility, we characterize a linear tradeoff between user and scheduler risk.
We study three canonical mechanisms:
(i) \emph{user-friendly}, which charges only for realized execution,
(ii) \emph{scheduler-friendly}, which charges for worst-case declared execution, and
(iii) \emph{Even-Steven}, which splits unrealized cost evenly.
We further analyze these mechanisms under three scheduler models: optimistic, pessimistic, and median, capturing different types of scheduler's utilities, when the execution of transactions is unknown.
Together, these results provide a principled design space for fee mechanisms under execution uncertainty.

\paragraph{Peril 2: Shill attacks on Gas Consumption Mechanisms.}
We build on the Gas Computation Mechanism (GCM) introduced by prior work~\cite{acilan2025transactionfeemarketdesign}.
A GCM assigns an abstract notion of \emph{effective compute} to each transaction given a schedule, while the TFM prices this compute based on user's priority.
Unlike prior work, we place this framework in an explicitly adversarial setting: both users and the scheduler are modeled as rational, economically motivated parties who may inject strategically crafted \emph{shill transactions}.

We identify two attack classes:
\emph{user shill attacks}, where fake transactions reduce the effective gas charged to a real transaction, and
\emph{scheduler shill attacks}, which inflate the gas charged to others.
We formalize \emph{shill-proofness} as a necessary security property for GCMs and show that it is fundamentally incompatible with the efficiency requirement that total gas equals the schedule makespan.
This tension explains why several natural GCM proposals admit profitable manipulation, and motivates fee-based notions of shill resistance that account for both pricing and execution outcomes.

We use one \emph{running example} throughout the paper. Alice sends a swap to an automated market maker that first reads an oracle price and only completes the trade (touching the pool and her balance) if the price is favorable; otherwise it reverts after the read. The set of objects it touches therefore depends on state it cannot see in advance, which makes it a contingent transaction. As a concrete user shill attack on this swap, suppose it touches a popular object and would normally be charged a high fee because it blocks other transactions from running in parallel. Alice can add a cheap, throwaway transaction that touches a \emph{different}, idle object. The mechanism now sees more ``parallel work'' in the block and spreads the cost out, so Alice's real transaction is charged less, even though her extra transaction did nothing useful. A scheduler shill attack is the mirror image: the scheduler adds its own fake transaction to make honest transactions look like they consume more, so it can collect more in fees.

\paragraph{Shill Attacks with Contingent Transactions.}
If the shill transaction does not have to pay the complete price for each object it accesses (i.e., scheduler risk is non-zero), then the cost to perform shill attacks reduces. Thus transaction contingency influences the shill proofness of a protocol. 

\paragraph{Object-Weighted Transaction Fee Mechanism.}
To translate gas consumption into payments under parallel execution, we introduce the \emph{Object-Weighted Transaction Fee Mechanism (OW-TFM)}.
Rather than charging purely based on serial execution time, OW-TFM prices a transaction according to the \emph{objects it declares in the parallel schedule}, including both compute time and contention over shared objects.
Intuitively, a transaction pays more not only when it runs longer, but also when it limits parallelism by using more objects.
Formally, OW-TFM prices each object based on how often the object has been used in previous blocks, in a EIP-1559-like~\cite{eip1559} price update function. 
OW-TFM then charges a fee equal to the sum of price of each object multiplied by the user-specified priority fee.
We show parameters under which this price function leads to a shill-proof system.

\paragraph{Contributions.}
In summary, we contribute the following:
\begin{enumerate}
\item We show that transactions with conditional execution can cause either users to overpay for computation objects that are never used, or the scheduler to lose revenue due to contingent transactions underexecuting and thus underpaying. We establish an inherent trade-off in fee design: any attempt to reduce user risk (overpayment) increases scheduler risk (uncollected fees), and vice versa.
\item We introduce a formal framework to model parallel transaction execution with contingency for analyzing TFMs.
\item We identify user-level and scheduler level \emph{shill attacks} where malicious party (user or scheduler) can add fake transactions to adversarialy decrease or increase fees.
\item We propose new transaction fee mechanism properties (\emph{shill proofness}) and parameters spanning the user-scheduler risk spectrum, and present a variant of weighted area TFM that can satisfy these properties.

\end{enumerate}

\paragraph{Outline.}
We begin by going over related work in \cref{sec:relatedwork}.
Then, we present our model in \cref{sec:model}.
In \cref{sec:contingenttx}, we analyze how contingency in transactions can cause either overpayment or loss in revenue depending on whether a user pays even if its transaction under-executes.
\cref{sec:newprotocol} maps the decision choices a designer must make when dealing with uncertain execution of transactions.
Independently, we observe that previous works developing TFMs dependent parallel execution are vulnerable to shill attacks (\cref{sec:shillattacksgcm}).
In \cref{sec:feeshills}, we see how contingent transactions make these attacks cheaper and easier to accomplish and define formal properties that every TFM design must satisfy.
In \cref{sec:tfm-design-space}, we look at object-weighted TFM (OW-TFM) which satisfies shill proofness.
We conclude with a discussion of design choices and future directions that can be taken in \cref{sec:discussions}.

\section{Related Work}
\label{sec:relatedwork}
\paragraph{Parallel Execution and Blockchain Scalability.}
The transition from serial to parallel execution is well-documented in recent systems research. Ponnapalli et al. identified execution as the primary bottleneck in modern blockchains, leading to the development of parallel engines~\cite{ponnapalli2021rainblock}. Prominent examples include Sui, which utilizes a fee-ordered DAG and object-based dependencies~\cite{suigas, sui2025object}, and Monad, which employs optimistic execution where conflicts trigger re-execution~\cite{Monad}. These systems rely on high-throughput consensus protocols like Mysticeti~\cite{babel2025mysticeti} and Bullshark~\cite{spiegelman2022bullshark} to order transactions rapidly before execution. While these systems improve throughput, they introduce new overheads regarding resource locking and contention that standard serial fee markets do not address.

\paragraph{Transaction Fee Mechanisms (TFMs).}
Traditional TFMs, such as the first-price auctions in Bitcoin and the base-fee burning mechanism of Ethereum's EIP-1559, treat execution resources as fungible and one-dimensional~\cite{gafni2024barriers,bahrani2024transaction, chung2023foundations,gafni2024discrete,garimidi2025transaction,roughgarden2024transaction,gafni2026transaction,gafni2026deterring}.
Recent literature has expanded this to \emph{multidimensional fee markets}, where distinct resources (e.g., storage and compute) are priced separately to avoid congestion in specific domains.
Diamandis et al. and Kiayias et al. explore the design of these multidimensional markets to better reflect resource scarcity~\cite{diamandis2023designing, kiayias2025one, acilan2025transactionfeemarketdesign}.
Our work builds on this by incorporating the execution time and parallelism into the pricing model, a necessity for the safe operation of parallel chains.

\paragraph{Gas Computation in Parallel Settings.}
Our work directly builds on the framework proposed by Acilan et al., which introduced Gas Computation Mechanisms (GCMs) like the Shapley and Time-Proportional Makespan (TPM) approaches~\cite{acilan2025transactionfeemarketdesign}. Their work established various properties which are good to have in GCM design, and propose various mechanisms to achieve (some of) the properties. We identify that these properties are insufficient to prevent \emph{Shill Attacks}. We demonstrate that their proposed mechanisms allow adversaries to manipulate gas costs by injecting \emph{fake  transactions}, that have no added utility from its execution.

\paragraph{Scheduling and Complexity.}
The problem of block building has historically been modeled as a Knapsack Problem, optimizing fee revenue under a gas limit~\cite{blockbuilding, mohan2024blockknapsack}.
In parallel environments, this transforms into a scheduling problem where transaction order dictates concurrency.
We analyze greedy heuristics used by systems like Sui, and random heuristics like Aptos and Monad, showing they are revenue-suboptimal compared to optimal parallel schedules. Furthermore, we ground our impossibility results in complexity theory, specifically the limits of parallel computation. We leverage standard results regarding the P-completeness of generic machine simulation to prove the hardness of verifying contingent object usage.

\section{Model}
\label{sec:model}
\paragraph{Transactions.}
A transaction $\tx(\compunits,\gasfee, \ins,\{\readset,\writeset\})$ is a (finite) sequence of instructions that interact with the blockchain state. The parameter $\compunits$ denotes the (upper bound on) compute units required to execute the transaction on a \emph{standard machine}.\footnote{Standard machine specifications, e.g., covering CPU and memory requirements, provide a fixed reference profile to measure $\compunits$.} Parameter $\gasfee$ denotes the price the user is willing to pay for all compute used. The instruction sequence ($\ins$) determines the control flow (conditionals, reads, writes) and how objects in the read and write sets (\readset, \writeset) are accessed. $\readset$ is the set of objects that the transactions accesses, but does not modify, whereas $\writeset$ represents the set of objects that it modifies. Since metadata beyond compute time and object access is often immaterial for scheduling, we represent transactions as
$
\tx \simeq (\compunits, \gasfee, \{\readset, \writeset\}). 
$
A transaction's declared read and write sets may be conservative: a transaction may declare \emph{contingent} objects that are only conditionally accessed at runtime. 

\paragraph{Schedule.}
A \emph{schedule} $\schedule$ is an ordered set of transactions that are to be executed in the given block, in which any parent of a transaction is executed strictly before the child. Typically, a schedule is limited by the amount of gas that all transactions in the schedule combined can consume. However, in this paper, we would constrain the schedule by the total compute consumed by all transactions (i.e., ignore other aspects of gas computation such as storage).
The schedule is generated from a publicly known algorithm, which we label as the \emph{scheduler}.

The revenue of the scheduler is the total \emph{fee} received from all the scheduled transactions. We represent the predicate that a transaction is in the schedule by $\tx \in \schedule$. We use the term prefix schedule to define a schedule until a particular transaction, i.e., $\schedule_{\tx}$, which represents the schedule with all transactions that must be executed before the transaction $\tx$.

\paragraph{State and Execution.}
Let $\state$ denote the global blockchain state (collection of object values) at a point in time. Given a schedule $\schedule$, let $
\state \xrightarrow{\schedule} \state'
$
denote the execution of the transactions in $\schedule$ (in order) on state $\state$, producing state $\state'$.
The \emph{execution outcome} of running a transaction $\tx$ on state $\state$ is written as
$\exec(\state,\tx) =\outp ,$
where $\outp$ is the result of executing $\tx$.

\paragraph{Objects.}
An object $\object$ can be thought of as data that each transaction reads, modifies, or creates on the blockchain. The set of all objects is represented as $\setobjects$, and $\numobjects = |\setobjects|$.
Each transaction $\tx(t,g,\{\readset,\writeset\})$ interacts with the objects in the sets \readset and \writeset. For all objects in \readset, the object remains unchanged after the interaction. For all objects in \writeset, the object can either be created or, if it already exists, then the value is modified. We represent the value of the object $\object$ by \val{\object}.
Sometimes a transaction might not access all the objects it asks for.

Objects have been used before by various blockchains, e.g., SUI Packages are smart contracts and functions (a type of read-only objects) while SUI Objects represent storage that each transaction can access \cite{sui2025object}.

We introduce a single \emph{running example} that we reuse throughout the paper to make each new idea concrete.

\begin{example}[Running Example: Alice's Swap]
Alice submits a transaction $\txswap$ to an automated market maker (AMM)~\cite{amm} which first \emph{reads an oracle price} object $\opx$.
If the price is favorable, it executes the swap, writing the pool reserves $\opool$ and Alice's balance $\obal$; if the price is unfavorable, it reverts right after reading the oracle and touches nothing else. So $\txswap$ declares the objects $\{\opx, \opool, \obal\}$, but only $\opx$ is \emph{always} used; $\opool$ and $\obal$ are \emph{contingent}, used only on the favorable branch. Throughout, we assume $\txswap$ uses one unit of compute, has priority $1$, and that $\opool$ is used thrice as much as other objects.
We call the favorable branch the \emph{desired} execution, and the revert the \emph{under-executed} one.
For this, we define contingency and some related terms below.
\end{example}

\paragraph{Used Sets.}
As a transaction's declared read/write sets may be conditional, we define the \emph{used} sets of $\tx$ when executed on a particular state (or after a prefix schedule) as
$\usedread(\state,\tx) \subseteq \readset(\tx)$
and
$\usedwrite(\state,\tx) \subseteq \writeset(\tx)$,
i.e., the subset of declared objects actually accessed during execution. When $\tx$ is run in the context of a schedule $\schedule$ applied to initial state $\state$, the used sets are dependent on the prefix $\schedule_{\tx}$  and we thus write $\usedread(\state,\schedule_{\tx},\tx)$ and $\usedwrite(\state,\schedule_{\tx},\tx)$ when clarity is required.

\paragraph{Contingent Objects.}
Intuitively, an object is contingent for $\tx$ if depending on prior transactions, the transaction conditionally accesses the object.

\begin{definition}[Contingent Object]
A declared object $\object\in \readset(\tx)\cup\writeset(\tx)$ is a \emph{contingent object} for $\tx$ from the perspective of some initial state $\state$ if there exist two legal prefix schedules drawn from the possible transaction sets $\schedule_{\tx}$ and $\schedule'_{\tx}$ such that
$\state \xrightarrow{\schedule_{\tx}} \state'$,
$\state \xrightarrow{\schedule'_{\tx}} \state''$,
where either the read used set differs (i.e., $\object\in\usedread(\state,\schedule_{\tx},\tx)$ but $\object\notin\usedread(\state,\schedule'_{\tx},\tx)$), or the write used set differs (meaning: $\object\in\usedwrite(\state,\schedule_{\tx},\tx)$ but $\object\notin\usedwrite(\state,\schedule'_{\tx},\tx)$).
\end{definition}

\paragraph{User Declarations.}
We consider the possibility of users declaring contigent object sets in their transactions, and denote these sets by
$\contingentread(\tx) \subseteq \readset(\tx)$ and $\contingentwrite(\tx) \subseteq \writeset(\tx)$.

\paragraph{Contingent Transactions.}
A transaction is contingent iff at least one of its declared objects is a contingent object, which can be formalized as follows.
\begin{definition}[Contingent Transaction]
A transaction $\tx$ is a \emph{contingent transaction} iff there exist two prefix schedules $\schedule_{\tx}$ and $\schedule'_{\tx}$ such that
$\state \xrightarrow{\schedule_{\tx}} \state', \state \xrightarrow{\schedule'_{\tx}} \state'',$
and the objects accessed during $\tx$'s execution on these two outcome states differ:
\[
\usedread(\state',\tx) \neq \usedread(\state'',\tx) \; \vee  \;\usedwrite(\state',\tx) \neq  \usedwrite(\state'',\tx)
\]
\end{definition}

Of these outcomes, one corresponds to the user's desired execution.
This outcome is assumed to use the contingent objects.
The other \emph{under-executed} outcome do not use the contingent objects and potentially uses less compute.
Note that contingency is about \emph{which} objects a transaction ends up touching, not \emph{what values} those objects hold.
A transaction that always reads and writes the same objects is \emph{not} contingent, even if an earlier transaction changed the values stored in those objects (e.g., copying a variable $a$ modified by a prior transaction).
It becomes contingent only when earlier transactions can change the set of objects it actually accesses.

\begin{example}
We now provide several examples of contingency.
\begin{itemize}
    \item A transaction that conditionally reads an oracle price and only writes balances if the price exceeds a threshold is contingent: the oracle object is a contingent object because prior transactions may have changed its value.
    \item A transaction that submits a swap to router which redirects the swap to the AMM that offers the best price is contingent.
    All AMMs' liquidity are contingent objects.
    However, in this transaction, the scheduler should expect the transaction not to access all objects.
    Such examples are out of scope and left as future work (see \cref{sec:discussions}).
    \item A pure transfer that reads only its own nonce and always writes the same set is non-contingent if no declared object can vary the control flow.
\end{itemize}
\end{example}

\paragraph{Transaction Conflicts.}
As in prior works, transactions $\tx_1$ and $\tx_2$ are said to \emph{conflict} if they access a common object and at least one access is a write. Formally,
\begin{align*}
\conflicts{\tx_1}{\tx_2}
\Longleftrightarrow
\{(\readset(\tx_1)\cap&\writeset(\tx_2)) \;\cup\; (\writeset(\tx_1)\cap\readset(\tx_2))\\
&\;\cup\;(\writeset(\tx_1)\cap\writeset(\tx_2))\} \neq \emptyset.
\end{align*}
If $\conflicts{\tx_1}{\tx_2}$ then $\tx_1$ and $\tx_2$ cannot be executed in parallel.
Contingency interacts with conflicts: if an object is contingent and ends up unused in a given execution, the conflict may not materialize for that run; thus the conflict relation may be \emph{execution-dependent} when contingent objects are present.

\paragraph{Fees and Effective Compute.}
Each transaction specifies a \emph{maximum fee} and a \emph{compute price} (price per unit of effective compute\footnote{Effective compute is equivalent to gas used, however, only the compute part of gas consumption is useful for us and thus in accordance to the multi-dimensional fee proposals, any other fee is separated.}).
While traditional serial blockchains equate effective compute with execution time ($\compunits$), parallel execution introduces a dependency: a transaction's cost should depend on its ability to run concurrently with the rest of the schedule.

\paragraph{Modeling Contingency.}
Contingent transactions introduce uncertainty: declared objects may not be used at runtime. To analyze the gap between declared and realized execution, we define the following definitions for fees for a transaction $\tx$ executed on state $\state$ under schedule $\schedule$. To cover fee burning, we also distinguish between the \emph{fee paid} by the user ($f$) and the \emph{fee received} by the scheduler ($r$). We use $f$ and $r$ as the generic paid and received fees; subscripts pick out specific reference points. In particular, $\attainablefee$, $\basefee$, and $\uifee$ (defined below) are hypothetical fees used for comparison, while $\paidfee$ and $\receivedfee$ are the \emph{actual} (realized) amounts paid by the user and received by the scheduler for a given execution.

\begin{itemize}
    \item \textbf{Attainable Fee ($\attainablefee$):} The fee $\tx$ would pay if it declared all contingent objects as definitively used. This represents the maximum theoretical cost (no uncertainty).
    \item \textbf{Baseline Fee ($\basefee$):} The fee $\tx$ would pay if it declared only its deterministic (guaranteed used) objects. This represents the minimum guaranteed cost.
    \item \textbf{User-Ideal Fee ($\uifee$):} The fee corresponding to exactly the objects $\tx$ actually uses during execution ($U_{used}$).
    \item \textbf{Actual Paid Fee ($\paidfee$) and Received Fee ($\receivedfee$):} The fee actually charged to the user and the portion received by the scheduler after burning, respectively.
\end{itemize}

We impose two monotonicity constraints on these fees:
\begin{enumerate}
    \item $\basefee \le \paidfee$: Users cannot reduce fees below the baseline cost of their guaranteed execution path.
    \item $\paidfee \le \attainablefee$: Declaring uncertainty (contingency) should not cost more than declaring certainty.
\end{enumerate}

\begin{example}[Running Example: The Four Fees for Alice's Swap]
Recall $\txswap$ declares $\{\opx,\opool,\obal\}$, say with object costs $1,3,1$.
The attainable fee treats all declared objects as used, so $\attainablefee = 1+3+1 = 5$.
The baseline fee counts only the always-used oracle, so $\basefee = 1$.
The user-ideal fee depends on the branch taken:
in the desired (favorable) execution Alice touches everything, so $\uifee = 5$;
in the under-executed (revert) branch only the oracle is touched, so $\uifee = 1$.
The interesting case is the under-executed one, where $\basefee = \uifee = 1$ sit far below $\attainablefee = 5$: the question of who pays for the $4$ units of unused, but reserved, capacity is discussed in \cref{sec:contingenttx}.
\end{example}

\paragraph{Scheduler Incentives and Burning.}
In mechanisms like EIP-1559~\cite{eip1559}, some of the paid fee can be burned.
We denote the \emph{retention ratio} as $\gamma = r/f$. To prevent schedulers from manipulating execution outcomes (e.g., forcing transactions to fail or under-execute to minimize work while collecting fees), we introduce a dynamic burning constraint: the retention ratio must increase as the transaction utilizes more of its declared objects. For simplicity, we assume that $\gamma$ is constant. 

This assumption aligns with practice.
For example, EIP-1559 defines for each block some fixed base fee which is burnt from each included transaction.

\subsection{Adversarial Model}
\label{sec:adversarial-model}

We model all parties as rational and economically motivated.  
There are two strategic roles: (i) \emph{users}, who submit transactions, and  
(ii) the \emph{scheduler}, which selects an execution schedule \(\schedule\).  
Consensus is assumed honest and only determines the candidate transaction set $\consensusset$; however a user or scheduler can insert arbitrary transactions in $\consensusset$.

We assume that all parties know the transaction fee mechanism (TFM), gas computation mechanism (GCM), and the risk-division rule that would be defined in~\cref{sec:newprotocol}.  
They can form beliefs about the outcome of contingent executions but cannot predict on-chain execution of transactions.  

\paragraph{User Actions.}
A user may:
(a) mis-declare contingent objects to influence \(\uifee\) or \(\attainablefee\);
and  
(b) submit \emph{shill transactions} which lower the total fee paid.
As a user's valuation $\val{\tx}$ is private, each one acts to maximize its expected utility \(U_{\text{user}} = \val{\tx} - \paidfee\) by minimizing $\paidfee$. 

\paragraph{Scheduler Actions.}
The scheduler may:
(a) insert shills if profitable, or  
(b) schedule transactions with favorable contingency risk.  
Its utility equals the total collected fees:
\(\schedutility = \sum_{\tx \in \schedule}(\receivedfee(\state,\schedule, \tx))\).
Since \receivedfee is a constant fraction of \paidfee, we instead consider the scheduler's objective in terms of \paidfee.
Thus, the scheduler's goal is to maximize the \emph{total} \paidfee summed over all transactions it includes in the schedule, rather than the fee of any single transaction in isolation.

\section{Contingent Transactions}
\label{sec:contingenttx}
We now turn our focus to \emph{contingent transactions}, that is, transactions that do not execute as expected.

\subsection{User Risk and Scheduler Risk Tradeoff}
In what follows, we prove that contingency poses a risk to users and schedulers: unless the scheduler has exponential compute, a protocol can never mitigate both types of risk at the same time.
To quantify the tradeoff, we define these risks.

The user risk due to contingency entails paying for unused declared objects, as we now formalize.

\begin{definition}[User Risk]
\label{def:userrisk}
For a transaction $\tx$ executed in schedule $\schedule$ from initial state $\state$, the \emph{user risk} is
\begin{equation}
\userrisk(\state,\schedule,\tx) \;:=\; \paidfee(\state,\schedule,\tx) - \uifee(\state, \schedule, \tx).
\end{equation}
That is, the user risk is the amount the user actually pays in excess of the user-ideal fee. 
\end{definition}

However, if contingent objects are not paid for, then contingent transactions do not pay their attainable value.
This reduces the scheduler's utility and thus also increases its risk.
As a reminder, paid fee is a proxy for the scheduler's utility, and this extends to the definition of risk.

\begin{definition}[Scheduler Risk]
\label{def:schedulerrisk}
For a transaction $\tx$ scheduled under $\schedule$ from initial state $\state$, the \emph{scheduler risk} is
\begin{equation}
    \schedrisk(\state,\schedule,\tx) \;:=\; \attainablefee(\state,\schedule,\tx) - \paidfee(\state,\schedule,\tx).
\end{equation}
Thus the scheduler risk measures the shortfall between the fee the scheduler expected in the best configuration and the fee actually collected.
\end{definition}

Our two risk metrics show the difference between protocols that make a user pay for all declared objects and those that require payment only for used objects.

\begin{lemma}[Risk Tradeoffs]
\label{lem:risk-tradeoff}
For any transaction $\tx$, initial state $\state$ and schedule $\schedule$, then the sum of $\userrisk$ and $\schedrisk$ is constant, independent of the execution of the transaction. 
\end{lemma}

\begin{proof}
By definition, $\userrisk+\schedrisk
= \paidfee - \uifee + \attainablefee - \paidfee$, so:
\[
\userrisk(\state,\schedule,\tx) + \schedrisk(\state,\schedule,\tx)
\;=\; \attainablefee(\state,\schedule,\tx) - \uifee(\state,\schedule,\tx).
\]
\end{proof}

By~\cref{lem:risk-tradeoff}, the sum of user and scheduler risks is constant. Thus, any protocol that reduces user risk must increase scheduler risk for that transaction (and vice-versa), unless the attainable fee is the same as the user ideal fee ($\attainablefee = \uifee$).

We note one subtlety: even though $\userrisk + \schedrisk$ is fixed, the \emph{combined} utility of the user and the scheduler is not. The reason is burning. Because the scheduler only keeps a fraction $\gamma$ of what the user pays (the rest is burned), money shifted from the user to the scheduler is not a clean one-to-one transfer. We make this dependence on $\gamma$ explicit when we analyze the mechanisms in \cref{sec:newprotocol,sec:feeshills}.

\begin{example}[Running Example: Splitting the Risk]
Take Alice's swap on its under-executed branch, where $\attainablefee = 5$ and $\uifee = 1$, so the total risk to divide is $\attainablefee - \uifee = 4$, regardless of how we price it. Writing $\paidfee = \alpha\,\attainablefee + (1-\alpha)\,\uifee$ for a knob $\alpha \in [0,1]$ (detailed in \cref{sec:newprotocol}), the split is: at $\alpha = 0$ Alice pays $\paidfee = 1$, so $\userrisk = 0$ and $\schedrisk = 4$ (the scheduler eats the wasted capacity); at $\alpha = 1$ Alice pays $\paidfee = 5$ for objects she never touched, so $\userrisk = 4$ and $\schedrisk = 0$; at $\alpha = \tfrac12$ she pays $\paidfee = 3$, splitting the risk evenly ($\userrisk = \schedrisk = 2$). In every case $\userrisk + \schedrisk = 4$: the trade-off only decides \emph{who} bears the $4$ units, never makes them disappear.
\end{example}

Intuitively, if the scheduler cannot determine the execution outcome (the used sets and observable outputs) of transactions when it forms a schedule, then it cannot simultaneously set the $\paidfee$ to be close to the user ideal fee ($\uifee$) and the maximum fee it can yield ($\attainablefee$). We formalize this intuition.

\begin{lemma}
\label{lem:cannot-determine-usage}
If the scheduler does not execute transactions while constructing the schedule, then in the general case it cannot decide whether some contingent transaction $\tx$ will access a contingent object $\object$ in its execution.
\end{lemma}
The proof follows from the fact that if the schedule has not been executed, then the state $\state'$ at which the contingent transaction is executed cannot be distinguished from whether it uses its contingent object or not. Formally,
\begin{proof}
    Let $\state'$ and $\state''$ be states such that $\exec(\state',\tx) \neq \exec(\state'',\tx)$ (by the definition of contingent transactions, such states much exist). Since the state before execution is not determined, it is unknown which execution route the transaction would take, and thus it is unknown if object $\object$ is accessed.
\end{proof}

\begin{theorem}
\label{thm:impossible-simultaneous-mitigation}
Let $\gaslimit$ denote the block compute limit used to validate schedule feasibility.
If a scheduler is constrained to decide schedules with at most $\gaslimit$ compute available (i.e., cannot perform more compute than the schedule itself), then for a schedule containing at least two conflicting contingent transactions, it is impossible to simultaneously set payments so that $\userrisk(\cdot)=0$ and $\schedrisk(\cdot)=0$ for all possible transaction and state inputs, except in a constant fee structure where the fee paid is independent of the objects used.
\end{theorem}

\begin{proof}
Suppose that the scheduler generates a schedule $\schedule$ guaranteeing $\userrisk(\state,\schedule,\tx)=0$ and $\schedrisk(\state,\schedule,\tx)=0$ for every $\tx,\state$. By~\cref{lem:risk-tradeoff}, we would have for every $\tx$ that $\paidfee(\state,\schedule,\tx) = \uifee(\tx)$ and $\paidfee(\state,\schedule,\tx)=\attainablefee(\tx)$.
Therefore, $\attainablefee(\tx)=\uifee(\tx)$ for every $\tx$. 
This implies that a transaction declaring less objects ($\uifee$) pays the same fee as when declaring a super-set of the objects ($\attainablefee$).
Therefore, the fee paid is independent of the objects it accesses.  

To ensure both risks are zero, the scheduler needs to know which contingent objects will be used (to charge exactly the user ideal fee), and adjust its schedule accordingly.
By~\cref{lem:cannot-determine-usage}, this entails executing transactions. 

Thus, the scheduler has a $\gaslimit$ computation available to execute transactions and its goal is to compute a schedule within $\gaslimit$ computation limit.
To do so, the scheduler must at least execute each transaction in the schedule, which implies a computation greater than $\gaslimit$. 
\end{proof}

\begin{remark}
    \cref{thm:impossible-simultaneous-mitigation} shows that when faced with conflicting contingent transactions, simultaneously eliminating user and scheduler risk is only possible in three cases.
    The first is the ``trivial'' solution: Only include non-conflicting transactions and thus avoid any contingency.
    The second is the constant fee structure, i.e., all transactions pay the same fee. Lastly, we can burn all the revenue that the scheduler makes, while requiring the user to pay only for the objects it uses.
\end{remark}

\subsection{Computational Impossibility}
Our impossibility is an \emph{economic} statement: under the budget $\gaslimit$, the scheduler cannot zero out both risks at once.
The natural follow-up question is \emph{why} the scheduler cannot spend some extra effort to figure out which contingent objects get used and price them exactly.
The rest of this section answers that with a \emph{computational} argument: deciding object usage is essentially as hard as running the transaction itself, so there is no cheap shortcut.
This is what makes the trade-off fundamental rather than an artifact of a weak scheduler, since no scheduler can predict usage without execution.

\paragraph{Object Usage Machine and Complexity.}
To make the hardness statement precise, we prove that there are no ``short-cuts'' to verifying contingency in general case.
Put differently, asserting whether an object declared by a transaction is used is sequential, when considering a general transaction specified in a Turing-complete language.
We proceed with the required technicalities.
\cref{def:OUM} formalizes an abstract machine that determines whether a transaction uses a given object.
\begin{definition}[Object Usage Machine (OUM)]
    \label{def:OUM}
    Given transaction $\tx$, an initial state $\state$, an object $\object$, and a gas limit $\gaslimit$, the \emph{Object Usage Machine} executes $\tx$ on $\state$, and, if $\object$ is accessed before $\gaslimit$ gas is used, then the machine halts and outputs $\accept$.
    Otherwise, the machine halts when the gas budget is exhausted and outputs $\reject$.
\end{definition}

We define a natural decision problem using \cref{def:OUM}.
\begin{definition}[Object Usage Decision Problem (OUDP)]
\label{def:OUDP}
Given $\left(\object,\tx,\state,\gaslimit\right)$ where the gas limit $\gaslimit$ is specified in unary, does the OUM output $\accept$?
\end{definition}

If a method to \emph{efficiently} decide \cref{def:OUDP} on any general input exists, then it could be utilized by schedulers to quickly verify contingency.
To understand what efficiency means in our case, note that for the performant protocols considered, block-time is strictly lower than the highest possible block execution time given the protocol's per-block gas limit and its reference hardware specifications.
Thus, the minimal efficiency requirement entails \emph{at least} verifying contingency in sub-polynomial time with respect to gas.

However, we show in \cref{res:OUDPsComplete} that OUDP is $\classP$-complete.
Under widely believed complexity class separations \cite{greenlaw1995limits}, this implies that deciding a general transaction's object usage at execution step $t$ requires sequential work proportional to $t$.
In turn, this means the user-scheduler risk trade-off is unavoidable if transactions can execute Turing-complete code.

\begin{theorem}[$\classP$-completeness of OUDP]
\label{res:OUDPsComplete}
The Object Usage Decision Problem (OUDP) is $\classP$-complete.
\end{theorem}
\begin{proof}
To show completeness, we need to prove that OUDP is in $\classP$ and that it is $\classP$-hard.
First, recall the definition of $\classP$.
\begin{definition}[Class $\classP$]
    The class $\classP$ equals the set of all decision problems (i.e., that produce a single output bit for any input) solvable by a deterministic Turing machine in polynomial time in input length.
\end{definition}

Now, we can prove that indeed OUDP is in $\classP$.
Given an instance $\left(\object,\tx,\state,\gaslimit\right)$ where the gas limit $\gaslimit$ is specified in unary, then one can simulate the execution of $tx$ on the OUM step-by-step until reaching the limit.
This takes polynomial time in the input size and thus that OUDP is in $\classP$, as $\gaslimit$ is encoded in unary.
Note that unary encoding is the literature's standard approach for such reductions \cite{greenlaw1995limits}.
Intuitively, that is to ensure that the running time permitted by $\gaslimit$ is polynomial in input length.
To see why, consider an input $\left(\object,\tx,\state,\gaslimit\right)$.
As $\gaslimit$ is a sub-string of this input, then the input's length is greater than or equal to the length of $\gaslimit$.
Given a binary encoding, then even the simplest step-by-step execution for $\gaslimit$ steps would formally require $O\left(2^\gaslimit\right)$ time, exponential in input length.
    
We continue by showing that OUDP is $\classP$-hard.
To do so, we reduce from the Generic Machine Simulation Problem (given in \cref{def:GMSP}), well-known to be $\classP$-complete \cite{greenlaw1995limits}.

\begin{definition}[Generic Machine Simulation Problem (GMSP)]
    \label{def:GMSP}
    Given a description $\bar{T}$ of a Turing machine $T$, a unary-encoded integer $n$, and an input string $x$, does $T$ accept $x$ within $n$ steps?
\end{definition}

Given a GMSP instance $\left(\bar{T}, n, x\right)$, our reduction proceeds by constructing an OUDP instance $\left(\object,\tx,\state,\gaslimit\right)$, thus tying OUDP and GMSP together.
Intuitively, we nest a universal Turing machine to simulate $T$ on top of the blockchain's VM.

A classic complexity theory result shows that there exists a universal Turing machine that can simulate the execution of any arbitrary machine \cite{greenlaw1995limits}.
Thus, set $\tx$ to equal the universal machine's implementation in the blockchain's VM that performs $n$ steps, and which, upon reaching an accepting state, accesses a fixed object $\object$ not involved in the implementation.
Note that the length of $\tx$ is constant and independent of the GMSP instance.
Next, we set the state $\state$ to equal the literature's canonical representation of GMSP instances as used for reductions, defined as follows.
Let $l$ be the input's length, $\#$ be a delimiter character which cannot be present in input strings, $p(l)$ be an upper bound on the running time of $T$ such that $p(l) = 2^{k_1 \lceil \log \left(l+1\right)\rceil}$ where $k_1$ is some integer power of two (traditionally chosen so that the encoding equals a turned-on bit followed by a trail of zeros, and thus is computable by bit-shifts \cite{greenlaw1995limits}), and define $\#^{p(l)}$ as the delimiter concatenated to itself $p(l)$ times.
Then, the canonical representation is $x\#\bar{T}\#^{p\left(l\right)}$.
Finally, given that $g_1$ is the maximal gas consumed to execute one step of the machine $T$ using the fixed machine encoded in $\tx$ and that $g_2$ is the gas needed to access $\object$, set $\gaslimit$ to $2^{k_2 \lceil \log \left(\left(g_1 \cdot n\right) + g_2\right) \rceil}$ for some integer power of two $k_2$.

In total, this is an efficient reduction such that $\tx$ accesses $\object$ within the limit $\gaslimit$ if and only if $T$ accepts $x$ within $n$ steps.
Combined with the $\classP$-completeness of GMSP, we get that OUDP is $\classP$-hard.
In turn, as OUDP is both in $\classP$ and $\classP$-hard, we conclude it is $\classP$-complete.
\end{proof}

\begin{remark}
    Our results show that the risk tradeoff is ex ante: the scheduler must reserve capacity and price uncertainty \emph{before} execution.
    This issue cannot be mitigated using ex post optimizations like pruning unused contingent objects, which require execution.
    Indeed, Ethereum's recent block-level access lists proposal makes a similar distinction \cite{AccessLists}.
    Thus, while the proposal requires recording touched objects when executing transactions at block creation time (to make later verification amenable to parallelization), it still emphasizes the possibility of adversarial ``phantom'' over-declarations.
\end{remark}

\section{Mapping the Scheduler Design Space}
\label{sec:newprotocol}
Given the fundamental impossibility of simultaneously minimizing both user risk (\userrisk) and scheduler risk (\schedrisk) when dealing with contingent transactions, we explore three fee division mechanisms that allocate this risk differently between the two parties.
Two of these correspond to the theoretical boundaries, that is, assign all risk to either users or schedulers, while the third balances risk equally.

\subsection{Risk Division}
\label{sec:riskdivision}
We present the \emph{user-friendly}, \emph{scheduler-friendly}, and \emph{Even-Steven} risk division mechanisms.
Each defines how paid fees (\paidfee) relate to attainable (\attainablefee) and user-ideal (\uifee) fees.

\subsubsection{User-Friendly Mechanism}

Under the \emph{user-friendly} fee division mechanism, users pay for exactly what they use, thus eliminating user risk entirely.
That is if the contingent transaction under-executes, then the user pays a fee that a transaction expecting this execution would have paid.
Formally, the paid fee equals the user-ideal fee: $\paidfee = \uifee$.
This implies that the scheduler absorbs all uncertainty associated with contingent execution, leading to
\[
    \schedrisk = \attainablefee - \paidfee = \attainablefee - \uifee, \quad \text{and} \quad \userrisk = 0.
\]
Intuitively, this is a user-oriented design: one only pays for the part of the transaction that actually executes.
But, the scheduler bears all risk, as its realized revenue decreases whenever transactions under-execute.

Using our formal notations, if the transaction executes completely, then $\uifee = \attainablefee$, and the scheduler receives the maximum possible fee.
When a transaction under-executes, however, $\uifee$ may decrease to the baseline fee $\basefee$.
Such a design is commonly seen in Ethereum-esque fee mechanisms, where parallelism is not a factor.

\subsubsection{Scheduler-Friendly Mechanism}
The \emph{scheduler-friendly} fee division lies at the opposite extreme: the scheduler bears no risk and receives the attainable fee regardless of execution outcomes: $\paidfee = \attainablefee$. 
This implies $\schedrisk = 0$ and $\userrisk = \paidfee - \uifee = \attainablefee - \uifee$.
Under this mechanism, the user must pay the full attainable fee \emph{even} if part of its transaction does not execute as intended.
The scheduler's revenue is thus insulated from contingency-related risks.
Such a design is proposed for chains like Monad~\cite{MonadGas}.

Scheduler-friendliness can be interpreted as an \emph{execution-agnostic} pricing rule: users always pay the highest possible fee and are not reimbursed post-execution for unused objects.
This guarantees revenue certainty for the scheduler, but exposes users to potential overpayment when contingent transactions do not make use of all declared objects.

\subsubsection{The Even-Steven Mechanism}

Between these two extremes lies the \emph{Even-Steven} mechanism, which evenly splits the risk between the user and the scheduler.
Under our notation, this can be formalized as requiring:
$\paidfee - \uifee = \attainablefee - \paidfee$.
Thus, solving for the paid fee yields:
\begin{equation}
\label{eq:evensteven}
    \paidfee = \frac{\attainablefee + \uifee}{2}.
\end{equation}
This equal-risk division ensures that the user and scheduler each bear half the uncertainty implied by contingency.

The Even-Steven mechanism can be viewed as a compromise: the user pays for half of the unrealized execution of the transaction.
If a transaction completely fails to execute, the user still pays halfway between the attainable fee and the user-ideal fee.
Conversely, if the transaction executes successfully, both user and scheduler pay the optimal amount.

\subsubsection{Generalized Representation}

The three mechanisms can be represented as points on a linear spectrum parameterized by $\alpha \in [0,1]$:
$\paidfee = \alpha \attainablefee + (1-\alpha)\uifee$.
Here, $\alpha$ encodes how much risk is ``transferred'' to the user:
\begin{itemize}
    \item $\alpha = 0$: \textbf{User-Friendly} (no user risk),
    \item $\alpha = \frac{1}{2}$: \textbf{Even-Steven} (equal risk),
    \item $\alpha = 1$: \textbf{Scheduler-Friendly} (no scheduler risk).
\end{itemize}

\subsection{Revenue Maximization Functions}
\label{sec:revenuescheduler}

The aim for each scheduler is to maximize the utility generated given contingent transactions.
For this, a scheduler can choose among utility functions that optimize different aspects of revenue, each operating based on a different expectation regarding the successful execution of contingent transactions.

The first scheduler maximizes revenue under the assumption of the most favorable possible execution outcome, utilizing the highest potential fee achievable.

\begin{definition}[Optimistic Scheduler]
    The \emph{Optimistic Scheduler} optimizes the best possible revenue under the assumption that \emph{all transactions successfully execute}.
\end{definition}

In the context of fees, if a transaction executes successfully and uses all declared objects, the user-ideal fee paid approaches the attainable fee ($\uifee = \attainablefee$). The optimistic scheduler maximizes revenue based on the expectation that the paid fee ($\paidfee$) equals the attainable fee ($\attainablefee$) for all scheduled transactions, thus optimizing based on the upper bound of potential income.
As scheduler expects $\uifee = \attainablefee$, all risk divisions lead to the same utility, $\paidfee^{s} = \attainablefee$.

The second scheduler minimizes risk exposure by planning for the worst-case revenue scenario.

\begin{definition} [Pessimistic Scheduler]
    The \emph{Pessimistic Scheduler} assumes the worst possible execution for all contingent transactions and \emph{optimizes the worst possible revenue}.
\end{definition}

A user would pay $\basefee$ for submitting a deterministic, non-contingent transaction with the same minimum computational and always-used object requirements. This fee $\basefee$ is a lower bound, since $\basefee \le \uifee$. Thus, for a scheduler expecting the worst possible fee from the transaction $\paidfee^{s} = \basefee$. 

From eq.~\cref{eq:evensteven}, the worst-case fee paid under the Even-Steven mechanism is $\paidfee = \frac{\attainablefee+\basefee}{2}$.
The pessimistic scheduler therefore optimizes revenue based on this minimum guaranteed fee that would be received even the transaction under-executes.

The third scheduler utilizes a probabilistic assumption to balance optimism and pessimism regarding contingency.

\begin{definition}[Median Scheduler]
    The \emph{Median Scheduler} is defined as one that assumes that half the contingent transactions will under-execute. This implies that, \emph{in expectation, it receives halfway between the potential revenue and minimum revenue} from the contingent transaction.
\end{definition}

Given the assumptions of this scheduler, it would expect the paid fee to equal
$
    \paidfee = \frac{3\attainablefee+\basefee}{4}
    .
$

We provide a summary of the different combinations of revenue maximization functions and risk divisions in \cref{tab:scheduler-utility}.
\begin{table}[H]
    \centering
    \caption{Scheduler utility under different risk divisions and execution outcome expectations.}
    \begin{tabular}{c c c c}
\toprule
& Optimistic &Pessimistic & Median\\
\midrule
     User-Friendly & $\attainablefee$ & $\basefee$ & $\frac{\attainablefee+\basefee}{2}$ \\
     Scheduler-Friendly&$\attainablefee$ &$\attainablefee$&\attainablefee\\
     Even-Steven &\attainablefee &$\frac{\attainablefee+\basefee}{2}$&$\frac{3\attainablefee+\basefee}{4}$\\
     \bottomrule
\end{tabular}
    \label{tab:scheduler-utility}
\end{table}

\section{Shill Attacks on GCMs}
\label{sec:shillattacksgcm}
We now analyze shill attacks on the GCMs presented in the literature and introduce properties required for a mechanism to be robust to these manipulations.
In such attacks, \emph{fake} transactions are introduced, either by the scheduler to increase the fee paid by users, or by users to reduce their expenses.

The specific attacks that we analyze target parallel-execution blockchains.
Previous works like~\cite{acilan2025transactionfeemarketdesign,diamandis2023designing,lavee2025does,kiayias2025one} note accurately pricing transactions requires considering that while two transactions' computation time might be the same, one may access more \emph{in-demand} objects than the other. 
In particular, Acilan et al.~\cite{acilan2025transactionfeemarketdesign} describe a set of important properties and suggest various GCM designs that satisfy subsets of these properties. 
We believe one property that was not considered, but should be of the highest importance, is shill-proofness.

As GCMs were first put forth by Acilan et al.~\cite{acilan2025transactionfeemarketdesign}, we follow their assumption that the fee paid per gas unit is the same for all transactions until otherwise specified, and further assume in this section that user payments are fully received by the scheduler. 
Shills are priced exactly like any honest transaction and must pay whatever the mechanism at hand prescribes, including any burned portion.
An attack is therefore only worthwhile when the gas (or fee) saved on the target transaction outweighs what the adversary must pay for shills, and such considerations are fully accounted for in our analysis.

\subsection{Attacking Existing GCMs}
We begin by demonstrating how shill attacks can be applied to two of the GCMs proposed by \cite{acilan2025transactionfeemarketdesign}. 
Let $\consensusset$ be the set of transactions scheduled, $t(\tx)$ the execution time of $\tx$, $O$ the set of objects accessed (analogous to the notion of storage keys considered by \cite{acilan2025transactionfeemarketdesign}), and $\makespan{S}$ the least amount of time taken to execute a subset $S \subseteq \consensusset$. 

\begin{definition}[Shapley GCM] Allocates gas used via Shapley value (marginal contributions to the makespan averaged over all subsets), that is,
    \[
        \gas{\consensusset}^{\text{Shapley}}(\tx) = \sum_{S \subseteq \consensusset \setminus \{\tx\}} 
        \frac{ \makespan{S \cup \{\tx\}} - \makespan{S}}{|\consensusset|\binom{|\consensusset|-1}{|S|}}
        .
    \]
    
\end{definition}

\begin{definition}[Time-Proportional Makespan (TPM) GCM]
    Assigns gas used proportional to the share of execution time within the block's makespan, i.e.,
    \[
        \gas{\consensusset}^{\text{TPM}}(\tx) = 
        \frac{ t(\tx) }{ \sum_{\tx' \in \consensusset} t(\tx') } \cdot \makespan{\consensusset}
        .
    \]
\end{definition}

In the two examples below, we assume $n \geq 2$ parallel threads and the optimal lock-based scheduler. Each transaction is written as $\tx \simeq (t, O)$, where $t$ is its execution time and $O$ is the set of objects it accesses.

\begin{example}[User Shill Attack on Shapley GCM]
    Let the set of other transactions be $\{\tx_1,\tx_2,\tx_3\}$ with $\tx_1 \simeq (1,\{\object_1\})$, $\tx_2 \simeq (1,\{\object_1\})$, $\tx_3 \simeq(2,\{\object_2\})$.
    A user's \emph{original} transaction is $\tx_4\simeq(3,\{\object_1\})$. $\consensusset = \{\tx_1,\tx_2,\tx_3,\tx_4\}$

    \paragraph{Baseline.}
    Under the Shapley GCM,
    $
        \gas{\consensusset}^{\text{Shapley}}(\tx_4) = \frac{8}{3}
        .
    $

    \paragraph{User Shill Attack.}
    If the user also submits a \emph{fake} $\tx_5\simeq(1,\{\object_2\})$, then:
    $
      \gas{\consensusset\cup\{\tx_5\}}^{\text{Shapley}}(\tx_4) = \frac{137}{60},\gas{\consensusset\cup\{\tx_5\}}^{\text{Shapley}}(\tx_5) = \frac{11}{30}
      .
    $
    Hence, the user's \emph{total} charge is $
      \frac{137}{60}+\frac{11}{30} = \frac{53}{20} < \frac{8}{3},
    $
    Thus, by sending two transactions, the user reduced the gas used.
\end{example}

\begin{example}[Scheduler Shill Attack on TPM GCM]
    Let the set of other transactions be $\tx_1 \simeq (6,\{\object_1\})$, $\tx_2 \simeq (6,\{\object_2\})$.

    \paragraph{Baseline.}
    Under the TPM GCM we have that the gas of both transactions is 3, i.e., $\gas{\consensusset}^{\text{TPM}}(\tx_1) = 3$ and $\gas{\consensusset}^{\text{TPM}}(\tx_2) = 3$.

    \paragraph{Scheduler Shill Attack.}
    If the scheduler submits a \emph{fake} $\tx_3\simeq(6,\{\object_1, \object_2\})$, then
    $
        \gas{\consensusset}^{\text{TPM}}(\tx_1) = \gas{\consensusset}^{\text{TPM}}(\tx_2) = \gas{\consensusset}^{\text{TPM}}(\tx_3) = 4
        .
    $
    Hence, the total gas received by the scheduler becomes $12-4 = 8 > 6$.
    Thus, by sending a fake transaction, the scheduler earns more gas. 
\end{example}

\subsection{Shill-proofness}
Taking a step back, we look at why these attacks happened and what properties are required in addition to those mentioned in~\cite{acilan2025transactionfeemarketdesign}.
For a user shill attack, the user was able to insert a fake transaction (or transactions) to reduce the gas used strategically. Denoting it in mathematical terms, if the set of transactions is $\consensusset$, and the user's transaction is $\tx_1 \in \consensusset$, then if there exists a $\consensusset'$, such that
\[
        \gas{\consensusset}(\tx_1) > \gas{\consensusset\cup \consensusset'}(\tx_1) + \gas{\consensusset\cup \consensusset'}(\consensusset'),
\]
then, sending the fake transactions $\consensusset'$ reduces gas usage. 

Based on this, we introduce the following property that the GCM must satisfy to prevent user shill attacks.

\begin{definition}[User Shill Proofness]
    \label{def:usershillproof}
    No user should be able to reduce the gas paid by their transaction by introducing more transactions, or $\forall\; \tx_1, \consensusset, \consensusset'$ then:
    \begin{align}
        \gas{\consensusset}(\tx_1) \leq  \gas{\consensusset\cup\consensusset'}(\tx_1) + \gas{\consensusset\cup\consensusset'}(\consensusset')
        .
    \end{align} 
\end{definition}

Consider the Set Inclusion property (Property 4, Acilan et al.~\cite{acilan2025transactionfeemarketdesign}), which states that if $T_1 \subseteq T_2$ then: $\gas{T\cup T_1}(T_1) \leq  \gas{T\cup T_2}(T_2)$.
User shill proofness is a subset of this property, i.e., if a GCM satisfies Set Inclusion, then it necessarily satisfies User Shill Proofness. 
However, the implications are much different. While Set Inclusion tries to counter user collusion, User Shill Proofness deals with users' fake transactions to reduce cost. While Set Inclusion is a good-to-have property, User Shill Proofness is a fundamental requirement to avoid throughput reduction due to fake transactions. We therefore view shill-proofness as a \emph{new} robustness notion rather than a restatement of coalition- or collusion-proofness for multidimensional fee markets~\cite{acilan2025transactionfeemarketdesign,kiayias2025one}. The gas-only Set Inclusion property asks how a single transaction's gas compares as its declared set grows; shill-proofness instead asks whether splitting a transaction into several (otherwise useless) transactions can lower the total charge. The gap widens once we move from gas to \emph{fees} in the presence of contingency (\cref{sec:feeshills}): our fee-based notions additionally account for the priority fee, the burned portion of the fee, and shills that under-execute or fail, none of which the gas-only property captures.

Next, we analyze the reasons behind scheduler shill attacks.
The scheduler can add a transaction in a natural GCM and increase the gas paid by the users, effectively increasing their revenue.
Denoting it mathematically, if the set of transactions submitted is $\consensusset$, then if there exists $\consensusset'$ such that $
    \gas{\consensusset}(\consensusset) <  \gas{\consensusset\cup \consensusset'}(\consensusset),
$
then the scheduler can add fake transactions in $\consensusset'$ to increase the total gas received by it. 

To prevent this, we introduce scheduler shill proofness.

\begin{definition}[Scheduler Shill Proofness]
    \label{def:schedshillproof}
    No scheduler should be able to increase the gas paid by honest transactions by introducing more transactions, or
    \begin{align}
        \gas{\consensusset}(\consensusset) \geq  \gas{\consensusset\cup \consensusset'}(\consensusset)
        .
    \end{align}
\end{definition}

Interestingly, we find that \nameref{def:schedshillproof} is at odds with a seemingly unrelated property, \nameref{def:efficiency}, which was suggested by Acilan et al.~\cite{acilan2025transactionfeemarketdesign}.
\begin{definition}[Efficiency]
    \label{def:efficiency}
    The sum of gas used by all transactions in the schedule must equal the makespan of the schedule, i.e., if $\makespan{\consensusset}$ represents the makespan of the schedule with transactions $\consensusset$, then $\gas{\consensusset}(\consensusset) = \makespan{\consensusset}.$
\end{definition}

\begin{restatable}[Scheduler Shill Proofness vs Efficiency]{lemma}{schedulershillefficiency}
    \label{lem:spamefficiency}
    It is impossible for a GCM with a limited number of parallel cores to achieve both Scheduler Shill Proofness and Efficiency.
\end{restatable}
\begin{proof}
Consider $n$ cores and $n+1$ completely parallel transactions, which take time $t$ on a single core to execute. Any subset of $n$ transactions can be scheduled in parallel such that the makespan is $t$, but the $(n+1)$-th transaction must be added sequentially, such that the makespan becomes $2t$. Scheduler Shill Proofness would require that the $(n+1)$-th transaction consume at least $t$ gas. Since all transactions could be the $(n+1)$-th transaction, all transactions must consume $t$ gas, and that violates the efficiency requirement that states the total gas consumed must be the makespan. 
\end{proof}

Note that the proof relies on atomicity of the transaction, and if the transaction execution can be distributed across multiple cores, then this proof would not work. This is because the sequential transaction can be divided into smaller chunks and passed between cores to complete the execution in $\frac{n+1}{n} \cdot t$ time.
This is further discussed in \cref{sec:discussions}.

While it may be tempting to conclude that if a Gas Computation Mechanism (GCM) satisfies the two properties defined above then it is resistant to shill attacks, thereby yielding throughput optimality (if the schedule is optimal), this conclusion is not correct.
In practice, most blockchain protocols also require users to specify a priority fee as discussed in \cref{sec:priorityschedules}.
Combined with the execution contingency as introduced in \cref{sec:contingenttx}, shill attacks can potentially be made much worse by the adversary. 
The scheduling objective proposed in \cite{acilan2025transactionfeemarketdesign}, which focuses on minimizing makespan, does not capture this priority structure.
We therefore argue that the objective should instead be to maximize utility for scheduler rather than minimize makespan.

The dependence on both fee and gas consumption is critical. A transaction may consume substantial gas while offering only a negligible gas price, creating opportunities for profitable shill attacks.
For instance, suppose a transaction offers a gas price of only $\varepsilon \sim 0$, yet its gas consumed is determined by one of the GCM methods described in \cite{acilan2025transactionfeemarketdesign}.
A recurring feature of these methods is that adding a highly parallelizable transaction to the schedule reduces the effective gas consumption attributed to other transactions. A user can therefore submit a low-priced, highly parallelizable fake transaction to decrease the gas consumption of their actual transaction, thereby lowering the total fee paid. 

To address this vulnerability, we follow by refining the notion of shill-proofness (and related properties), formulating them in terms of paid fees rather than raw gas consumption.

\section{Fee-based Shill Proofness}
\label{sec:feeshills}
We have so far analyzed shill attacks on GCMs through the lens of \emph{gas consumed}.
Now, we revisit shill attacks with a focus on the \emph{fee paid} in the presence of \emph{contingent transactions}. When a contingent transaction under-executes, it could pay a much lower fee compared to the fee it \emph{could} have paid if it had executed fully. This asymmetry enables new fee-reduction strategies: rational users can inject failing or near-failing transactions to alter contention and scheduling so that their high-value transaction pays strictly less. We formalize fee-based shill-proofness, show why baseline designs typically violate it when failures pay near-zero, and outline conditions that restore robustness.

We begin by first discussing the goal of the next two sections.
When considering how to compute gas (GCM) or the fee paid by the transaction (TFM) in the context of parallel execution, the first question that rises is, what schedule must be used?
Finding the best schedule $\schedule$ that optimizes scheduler revenue constitutes our first aim.
However, optimizing the revenue requires us to also know what fee a transaction pays in the schedule, and thus while computing the schedule, the fee paid by the transaction based on the said schedule and the priority fee set by the user, must also be computed in conjunction. Moving forward, we denote the schedule by $\schedule$, since given the fee-structure, all parties can deterministically calculate it.
Also, the sub-schedule is represented by $\schedule^{\consensusset}_\tx$.
Adversaries are assumed to not want to change the execution of any transaction, i.e., $\exec(\state \xrightarrow{\schedule^{\consensusset}_\tx}\state', \tx_1) = \exec(\state \xrightarrow{\schedule^{(\consensusset\cup\consensusset')}_\tx}\state'', \tx_1)$.
For brevity, denote $\schedule(\consensusset)$ by $ \schedule$, and $ \schedule(\consensusset\cup\consensusset') $ as $ \schedule'$.

\emph{User shill-proofness}~(\cref{def:usershillproof}) ensures that users would not be able to reduce the gas or computation units attributed to their transactions by adding shill transactions.
Given a model in which contingent transactions exist, an adversarial user must not be able to add \emph{fake} contingent transactions that always under-executes to reduce the \emph{fee} paid by its main transaction.  Formally,

\begin{definition}[Fee-Based User Shill Proofness]
\label{def:feebasedusershillproofness}
    Introducing fake transactions cannot reduce the total fee paid by the expected execution of transaction $\tx_1$, even if while under-executing the fake transactions have a zero user-ideal fee: $\uifee = 0$.
    That is, for any state $\state$ and $\forall\; \tx_1, \consensusset, \consensusset'$
    \begin{align}
        \paidfee(s,\schedule, \tx_1) \leq  \paidfee&(s,\schedule',\tx_1) + \paidfee(s,\schedule',\consensusset')
    \end{align} 
\end{definition}
Here, we assume that $\tx_1$ never under-executes and always pays the attainable fee $\attainablefee$. For the three risk divisions defined in~\cref{sec:newprotocol}, this translates to different conditions.
\begin{itemize}
    \item \textbf{User-Friendly Division} 
    \[\attainablefee(s,\schedule, \tx_1) \leq  \attainablefee(s,\schedule',\tx_1)
    \]
    \item \textbf{Scheduler-Friendly Division}
    \begin{align*}
    \attainablefee(s,\schedule, \tx_1) \leq  \attainablefee&(s,\schedule',\tx_1) +\attainablefee(s,\schedule',\consensusset')
    \end{align*}
    \item \textbf{Even-Steven Division} 
    \begin{align*}
    \attainablefee(s,\schedule, \tx_1) \leq  \attainablefee&(s,\schedule',\tx_1) +0.5\cdot\attainablefee(s,\schedule',\consensusset')
    \end{align*}
    \item \textbf{Generalized Division} 
    \begin{align*}
    \attainablefee(s,\schedule, \tx_1) \leq  \attainablefee&(s,\schedule',\tx_1)+\alpha\cdot\attainablefee(s,\schedule',\consensusset')
    \end{align*}
    \end{itemize}

Next, we look at re-defining scheduler shill proofness. Here, we relax the assumption that the scheduler utility is defined in terms of fee paid, since the scheduler is introducing its own transactions, and it only receives back the fee received back from, thus incurring a cost of $\paidfee-\receivedfee$ from the transaction. From~\cref{def:schedshillproof}, a scheduler must not be able to strategically insert fake transactions such that the fee paid by other transactions increases more than the cost it pays. 

\begin{definition}[Fee-Based Scheduler Shill Proofness]
\label{def:feebasedschedshillproofness}
    No scheduler should be able to increase the revenue received from honest transaction ($\tx_1$) by more than the cost incurred by introducing fake transactions ($\consensusset'$), even if all transactions in $\consensusset'$ under-execute. In other words, $\forall\; \tx_1, \consensusset, \consensusset'$
    \begin{align}
        \receivedfee(s,&\schedule',\tx_1)- \receivedfee(s,\schedule, \tx_1) \\ \notag
        \leq  &\paidfee(s,\schedule',\consensusset') - \receivedfee(s,\schedule',\consensusset')
    \end{align} 
\end{definition}

We assume that $\receivedfee = \gamma\cdot \paidfee$, and thus,
\begin{align}
    \gamma(\paidfee(s,\schedule', \tx_1) - \paidfee(s,\schedule,\tx_1)) 
    \leq  (1-\gamma)\paidfee(s,\schedule',\consensusset')
\end{align}

Again, we assume that $\tx_1$ never under-executes and always pays the attainable fee $\attainablefee$, and all introduced shill transactions under-execute to pay as little as possible.
For the risk divisions of~\cref{sec:newprotocol}, this translates to different conditions.
\begin{itemize}
    \item \textbf{User-Friendly} 
    \[\attainablefee(s,\schedule', \tx_1) -  \attainablefee(s,\schedule,\tx_1) \leq 0
    \]
    \item \textbf{Scheduler-Friendly} 
    \begin{align*}
    \attainablefee(s,\schedule', \tx_1) &-  \attainablefee(s,\schedule,\tx_1) \leq \frac{1-\gamma}{\gamma} \attainablefee(s,\schedule',\consensusset')
    \end{align*}
    \item \textbf{Even-Steven Division} 
    \begin{align*}
    \attainablefee(s,\schedule', \tx_1) &-  \attainablefee(s,\schedule,\tx_1) \leq \frac{1-\gamma}{2\gamma} \attainablefee(s,\schedule',\consensusset')
    \end{align*}
    \item \textbf{Generalized Division} 
    \begin{align*}
    \attainablefee(s,\schedule', \tx_1) &-  \attainablefee(s,\schedule,\tx_1) \leq \alpha\frac{1-\gamma}{\gamma} \attainablefee(s,\schedule',\consensusset')
    \end{align*}
    \end{itemize}

Based on the above feasible regions, we prove \cref{thm:independence-userfriendly}, which implies that transaction fees must not depend on parallelism if a user friendly mechanism is used.
\begin{theorem}[Independence in User-Friendly Mechanisms]
\label{thm:independence-userfriendly}
    In a parallel execution environment using a User-Friendly risk division (where User Risk $\userrisk=0$ and $\alpha=0$), a Transaction Fee Mechanism (TFM) satisfies both Fee-Based User Shill Proofness and Fee-Based Scheduler Shill Proofness if and only if the attainable fee of a transaction is independent of all other transactions in the schedule. 
\end{theorem}
The proof follows from the User-Friendly risk division's conditions for Fee-Based User and Scheduler Shill Proofness.

\begin{proof}
    Under user-friendliness, User Shill Proofness implies:
    \[
        \attainablefee(s,\schedule, \tx_1) \leq  \attainablefee(s,\schedule',\tx_1)\]
    And, scheduler shill proofness gives:
    \[
        \attainablefee(s,\schedule', \tx_1) \leq  \attainablefee(s,\schedule,\tx_1)\]
    The two can be satisfied together only when:
    \[
        \attainablefee(s,\schedule', \tx_1)=  \attainablefee(s,\schedule,\tx_1)\]
    And, this implies that the fee paid by the user is independent of other transactions in the system.
\end{proof}

Let's revisit the running example to illustrate shill attacks.
\begin{example}[Running Example: Shilling Alice's Swap]
The two attacks apply to Alice's swap.
In a \emph{user} shill attack, Alice adds a contingent transaction touching many idle objects, disjoint from $\{\opx,\opool,\obal\}$, but does not pay the fee since the transaction never uses the declared objects.
Because the throwaway runs in parallel with $\txswap$, the GCM or the TFM associated with it, now spreads the block's makespan over more transactions and attributes \emph{less} gas to $\txswap$, lowering Alice's fee.
In a \emph{scheduler} shill attack, the scheduler inserts a transaction that also touches the in-demand pool $\opool$; this makes $\txswap$ look like it contends for more shared state, inflating the gas charged to it. Both manipulations change Alice's fee without changing what her transaction actually does.
\end{example}

\section{Transaction Fee Mechanism Design Space}
\label{sec:tfm-design-space}
We now formalize the design space of TFMs for parallel execution environments and identify the key properties such mechanisms must satisfy. Our goal is to understand how different choices of \emph{schedulers} and \emph{risk-divisions} interact, and to characterize which combinations yield mechanisms that are both economically robust and resistant to manipulation.

At a high level, a TFM operates in conjunction with a scheduler that determines a feasible execution schedule for a set of transactions, subject to parallelism and object constraints. The fee mechanism must determine how execution uncertainty is priced, leading to different \emph{risk-division} regimes as defined in \cref{sec:riskdivision}, ranging from \emph{user-friendly} (users pay only for objects they actually consume), to \emph{scheduler-friendly} (users pay for all objects they declare), with intermediate compromises generalized by the generalized division rule. 

Further, when creating the schedule, the scheduler does not know the execution trace of the transaction, and thus cannot know the objects that would be used by the transaction. As defined in~\cref{sec:revenuescheduler}, the scheduler could have various utility definitions as well when designing the schedule, which ranges from \emph{optimistic}, \emph{pessimistic} as well as somewhere in between. We would look at how to create schedules based on only \emph{optimistic} and \emph{pessimistic} schedulers and leave a generalized utility function as a future work. 

A central difficulty in designing TFMs in this setting is the inherently \emph{cyclic} dependence between scheduling and fees. On the one hand, the scheduler's choice of which transactions to include and how to order or parallelize them depends on the fees offered by transactions. On the other hand, the fee paid by a transaction may itself depend on the realized execution schedule, particularly when fees are tied to the realized parallelism. This circular dependence complicates both the analysis and the implementation of TFMs.

\subsection{Object-Weighted TFM (OW-TFM)}

To break this cyclic dependency, we begin by focusing on a class of transaction fee mechanisms inspired by the weighted-area gas model introduced in prior work on transaction fee market design for parallel execution~\cite{acilan2025transactionfeemarketdesign}. In this model, the fee charged to a transaction is \emph{independent of the execution schedule}, while still capturing the effects of parallelism by pricing the objects based on past executed transactions. A key assumption for OW-TFM is that object demand and thus object prices, remain relatively stable across consecutive rounds, allowing transactions to be priced using per-object prices that are fixed at the time of inclusion.

An important consequence of using past transaction data to determine the fee paid by the transaction is that \emph{no shill attacks} are relevant unless the price is manipulated by inserting transactions in previous blocks.

Concretely, we assume that a transaction $\tx$ uses computation units $\compunits$ and declares a set of objects it may use during execution. Each object $\object$ has a price $p_\object$. If $\tx$ uses a object $\object$ and declares a priority fee of $\pi \geq 1$, then it pays a price $\pi\cdot p_\object\cdot \compunits$ for that object. To capture execution uncertainty, we parameterize the division of risk between users and the scheduler by a parameter $\alpha \in [0,1]$. For any object that is \emph{declared but not used} during execution, the transaction pays a fee of $\alpha \cdot p_\object$. Setting $\alpha = 0$ corresponds to user-friendly risk division, while $\alpha = 1$ corresponds to scheduler-friendly risk division. Intermediate values of $\alpha$ interpolate between these extremes and allow for more balanced risk-sharing arrangements.

\begin{example}[Running Example: Pricing Alice's Swap]
Under OW-TFM, with Alice's priority $\pi = 1$, compute $\compunits = 1$, and object prices $p_{\opx} = 1$, $p_{\opool} = 3$, $p_{\obal} = 1$, the favorable execution touches all three objects and costs $1 + 3 + 1 = 5$. On the under-executed branch she uses only $\opx$, but still owes $\alpha$ times the price of the two declared-but-unused objects, so she pays $1 + \alpha\,(3 + 1) = 1 + 4\alpha$. This recovers the same numbers as before: $1$ at $\alpha = 0$, $3$ at $\alpha = \tfrac12$, and $5$ at $\alpha = 1$, confirming that OW-TFM realizes the risk split of \cref{sec:newprotocol} with a fee that does not depend on the rest of the schedule.
\end{example}

We use a price update function similar to EIP-1559~\cite{eip1559}.
In order to formally describe this function, define the realized utilization of object $\object$ in block $\block$ (with transactions $B_\block$) as
\begin{align*}
\usage_{\object,\block}:=\sum_{\tx:=(\compunits,\_,{\readset,\writeset}) \in B_\block} \compunits\cdot\mathbf{1}\!\left\{\object \in \{\readset\cup \writeset\}\right\},
\end{align*}
which represents the time for which the object was reserved in the block . For each object $\object\in\setobjects$, we fix a target utilization $\usage_\object^\star>0$. Let $p_{\object,\block}>0$ denote the posted price for objects $\object$ in block $\block$. Fix a responsiveness parameter $\eta>0$.
The price for the next block is updated according to:
\begin{align*}
p_{\object,(\block+1)} &= p_{\object,\block}\cdot \exp\left(\eta \cdot \frac{\usage_{\object,\block}-\usage_\object^\star}{\usage_\object^\star} \right)\\
&\approx p_{\object,\block}\cdot \left(1+\eta \cdot \frac{\usage_{\object,\block}-\usage_\object^\star}{\usage_\object^\star} \right).
\end{align*}

\paragraph{Shill-proofness of the OW-TFM.}
Here, prices are set from the measured object utilization $\usage_{\object,\block}$ of the \emph{previous} block. Adding a transaction can only add to this measured usage or leave it unchanged: each declared object contributes a non-negative amount, and a contingent transaction that under-executes simply contributes less (possibly zero), but never a negative amount. So a shill can push the recorded usage of an object \emph{up} but cannot push it \emph{down}. This is why a contingent transaction, even though it may under-execute, cannot be used to \emph{lower} an object's measured consumption: under-execution only means the shill itself counts for less, not that honest usage is erased. The one way to reduce an object's recorded usage would be to remove honest transactions that use it (censorship), which is a separate kind of manipulation that we treat as out of scope. Consequently, the only relevant manipulation here is a scheduler shill attack that inflates an object's usage to raise its price in the next block.

\begin{lemma}[Scheduler Shill Proofness of OW-TFM]
    \label{lem:shill-ow-tfm}
    Consider block $B_\block$ with a block number $\block$ and a shill transaction $\tx^s$ which declares object $\object$ (i.e., $\object\in \readset\cup\writeset$).
    Let $\usage_\object^\star$ represent the object's target and average usage.
    If $\alpha \geq \pi_\object\eta\frac{\gamma}{(1-\gamma)}$, then the resulting TFM is scheduler shill proof, where $\pi_\object$ is the average priority for transactions using \object in block $\block+1$.  
\end{lemma}

\begin{proof}
    
    Since this is a multi-block attack, assume that the state is common before block $b$, and that $\tx_1$ is a transaction in block $\block+1$ (or all transactions that declare \object, combined into a single transaction), that has $\compunits_{\tx_1}$ of $\usage_\object$. The priority for the transaction is the average priority of all transactions in this block that use object $\object$ and is represented by $\pi_\object$.

    \paragraph{Without Shill Transaction.}
    In block $\block$, without the shill transaction, let the object utilization is $\usage_{\object,\block}$. The price of $\tx_1$ with priority $\pi_i$ without any manipulation in block $\block+1$ is:
    \begin{align*}
       \attainablefee(s,\schedule, \tx_1) = p_{\object,\block}\cdot \left(1+\eta \cdot \frac{\usage_{\object,\block}-\usage_\object^\star}{\usage_\object^\star} \right)\cdot \pi_\object\cdot \usage_\object^\star
    \end{align*}
    
    \paragraph{With Shill Transaction.}
    In block $\block$, let the scheduler add the shill transaction $\tx^s$. Since the priority would only increase the fee the shill transaction pays, we will assume that the priority fee is $1$. The cost of $\tx_1$ is now given by 
    \begin{align*}
        \attainablefee(s,\schedule', \tx_1) &= p_{\object,\block}\cdot \left(1+\eta \cdot \frac{\usage_{\object,\block}+\compunits_s-\usage_\object^\star}{\usage_\object^\star} \right)\cdot \pi_\object\cdot \usage_\object^\star \\
        &= \attainablefee(s,\schedule, \tx_1)+ \eta p_{\object,\block}\pi_\object\compunits_s,
    \end{align*}
    whereas the (attainable) fee paid by the shill transaction is 
    \begin{align*}
        \attainablefee(s,\schedule',\tx_s) = \pi_s\cdot p_{\object,\block}\compunits_s \geq p_{\object,\block}\compunits_s
    \end{align*}
    From fee-based scheduler shill proofness, 
    \begin{align*}
        \attainablefee(s,\schedule', \tx_1) -  \attainablefee(s,\schedule,\tx_1) &\leq \alpha\frac{1-\gamma}{\gamma} \attainablefee(s,\schedule',\tx_s)\\
        \eta p_{\object,\block}\pi_\object\compunits_s \leq \alpha\frac{1-\gamma}{\gamma}p_{\object,\block}\compunits_s &\implies
        \eta \pi_\object\frac{\gamma}{(1-\gamma)} \leq \alpha
    \end{align*}
\end{proof}

\subsection{Empirical Parameterization}\label{sec:empirical}
We now plug in real-world numbers from Ethereum to see what our theoretical bounds look like in practice.
 
\paragraph{Priority-to-base Fee Ratio ($\pi_o$).}
On Ethereum, the average priority fee is about $4\times$ the base fee, giving $\pi_o = 5$ (since $\pi_o$ includes the base fee itself). The distribution is skewed: most transactions have a low priority multiplier (median $\approx 1.2$), but a small fraction pay much more (the 95th percentile exceeds~$16$). These values are obtained from the Dune Analytics query in Open Science. We evaluate three scenarios: $\pi_o \in \{1.2,\; 5,\; 16\}$, representing median, mean, and high-priority transactions.
 
\paragraph{Responsiveness ($\eta$) and Retention Ratio ($\gamma$).}
EIP-1559~\cite{eip1559} sets $\eta = \frac{1}{8}$ as the price update speed.
The retention ratio $\gamma = \frac{r}{f}$ is the scheduler's fee share, with the rest burned.
Let $\gamma = \frac{1}{2}$ throughout, i.e., half of each fee is burned.
 
\paragraph{Scheduler Shill-proofness Threshold.}
By \cref{lem:shill-ow-tfm}, OW-TFM is scheduler shill proof when $\alpha \geq \pi_o \cdot \eta \cdot \frac{\gamma}{1-\gamma}$.
Plugging in $\eta = \frac{1}{8}$ and $\gamma = \frac{1}{2}$ gives $\alpha_{\min} = \frac{\pi_o}{8}$.
\cref{tab:alpha-threshold} summarizes the result for each scenario.
At the median priority, a small amount of user risk ($\alpha = 0.15$) is enough to make shill attacks unprofitable. At the mean priority, the required $\alpha$ rises to $0.625$, which falls between Even-Steven and scheduler-friendly. At the 95th percentile, $\alpha_{\min}$ exceeds~$1$, meaning no risk division within our framework can guarantee shill-proofness against the highest-priority attackers. However, since $\pi_o$ represents the \emph{average} priority across all transactions using a given object, the mean value of~$5$ is the most representative case for a typical block.
 
\begin{table}
\centering
\caption{Minimum risk-division parameter $\alpha$ for scheduler shill-proofness of OW-TFM under empirically-estimated Ethereum fee parameters.}
\begin{tabular}{lccc}
\toprule
Scenario & $\pi_o$ & $\gamma$ & $\alpha_{\min}$ \\
\midrule
Median priority & $1.2$ & $0.5$ & $0.15$ \\
Mean priority   & $5$   & $0.5$ & $0.625$ \\
95th percentile & $16$  & $0.5$ & --- \\
\bottomrule
\end{tabular}
\label{tab:alpha-threshold}
\end{table}

\paragraph{Contingency Impact.}
Heimbach et al.~\cite{heimbach2025dissecting} find that 42.6\% of Ethereum transactions would have benefited from access lists, and that 19.6\% of the access lists that were supplied turned out to be wrong, mostly because the state at execution time was not known in advance (which we prove in ~\cref{thm:impossible-simultaneous-mitigation}). These numbers show that contingency is common in practice. Under a scheduler-friendly mechanism ($\alpha = 1$), users with inaccurate declarations overpay for objects they never use. Under a user-friendly mechanism ($\alpha = 0$), the scheduler loses revenue instead. Since roughly 40\% of transactions are involved in conflicts~\cite{heimbach2024defi, wahrstatter2025execution} and a large share of those exhibit contingent behavior, the choice of~$\alpha$ has a real economic impact that grows with block size and object contention.

\section{Conclusion and Discussions}
\label{sec:discussions}
Parallel execution promises throughput gains for modern blockchains, but it fundamentally complicates the economics of transaction inclusion. In this paper, we show that when execution is uncertain, parallel execution based transaction fee mechanisms face inherent limitations. Execution contingency introduces a structural trade-off between user risk and scheduler risk: without pre-execution or collapsing to constant fees, no mechanism can simultaneously eliminate both. This impossibility is not an artifact of particular designs, but follows from the inherent sequentiality of determining object usage in expressive smart contracts. As a result, fee mechanisms must make explicit choices about how uncertainty is allocated, rather than attempting to hide it.

Independently, we demonstrate that parallel gas computation mechanisms are vulnerable to rational shill attacks by both users and schedulers, even when they satisfy previously proposed fairness and efficiency properties. We formalize fee-based shill-proofness as a necessary security requirement and show how execution contingency amplifies these attacks by reducing the cost of fake transactions. Taken together, our results clarify that pricing parallel execution is a mechanism-design problem with unavoidable trade-offs. Within this constrained design space, we show that object-weighted fee mechanisms with appropriate risk-division parameters can achieve robustness against shill attacks while remaining compatible with parallel execution. These findings provide foundational guidance for evaluating and designing transaction fee mechanisms in emerging parallel blockchain systems.

\paragraph{On Atomicity of Scheduled Transaction}
Using today's systems, the number of parallel processes can exceed the number of total cores that the system possesses.
This then leads to a distribution of workload across all cores through (very optimized) context switching of processes across all cores.
In such a case, the actual number of threads available for execution can be considered to be infinite, but the total time to execute each transaction is scaled by the extra number of threads in use at each time.
In such a case, consider the example in the proof for \cref{lem:spamefficiency}.
The $n+1$ transactions can now be scheduled on $n$ threads, and due to context switching all the threads would share the workload.
This leads to a makespan of $\sim \frac{n+1}{n} t$, which then can still satisfy efficiency, while satisfying scheduler shill-proofness.
This can be useful especially with unsafe mode, where the locks on objects are not checked at any moment. 
(Unsafe mode is safe since the schedule created never has shared overlapping objects)

\paragraph{On Attainable Fee.} We define the \emph{attainable fee} as the fee a transaction would pay if all declared contingent objects were used during execution. This abstraction simplifies analysis by providing a clean upper bound on what a scheduler might reasonably expect to collect from a transaction.

In practice, however, some transactions may never use all declared contingent objects under any execution path.
For example, consider an AMM router transaction that inspects multiple liquidity pools and conditionally routes execution to the one pool offering the best price.
While such a transaction may declare several pools as contingent objects, only one will be accessed.
In this case, the attainable fee does not correspond to any realizable execution.
We intentionally abstract away this nuance.
For the scheduler, distinguishing between ``logically impossible'' and ``unlikely'' execution paths is infeasible without executing the transaction.

\paragraph{On Execution of Transactions.}
Our framework treats execution outcomes adversarially or pessimistically. Incorporating probabilistic models of execution, informed by historical behavior or application-level semantics, could enable mechanisms that optimize expected risk rather than worst-case risk. Such approaches would necessarily introduce new assumptions and attack surfaces, and merit careful security analysis.

\paragraph{Declaring Contingent Objects.}
Our mechanisms require users to declare the objects that their transaction may touch, including contingent ones.
In practice, declarations can be produced by client software, the same way Ethereum access lists are generated by simulating transactions before broadcasting them.
Lying in either direction can be discouraged.
If a user \emph{under}-declares, i.e., touches an object it did not declare, this can be caught at execution time, since the runtime sees the access.
The transaction can then be reverted while still being charged, as already happens for failed Ethereum transactions.
If a user \emph{over}-declares, i.e., lists objects it never uses, this can be made costly by charging for declared-but-unused objects, which is exactly what the scheduler-friendly division (and any $\alpha > 0$) does.
Getting users to reveal their true object sets is an important deployment constraint, and a mechanism that fully enforces it is left to future work.
 
\paragraph{Relation to Ethereum and EIP-1559.}
Ethereum today uses EIP-1559~\cite{eip1559}, which prices a single, one-dimensional notion of gas and does \emph{not} charge transactions for contention over specific objects.
As long as execution is serial, there is no parallelism to price, so contention and contingency do not show up in the fee.
The point of this paper is that \emph{once} parallel execution is adopted, these effects appear and must be priced, otherwise the shill attacks and risk trade-offs we describe become possible.
This is why the setting is relevant to chains that already run transactions in parallel, such as Sui~\cite{suigas}, Solana~\cite{solanafoundation_2019}, and Monad~\cite{Monad}, and to Ethereum's own roadmap toward parallel execution via block-level access lists~\cite{AccessLists}.

\appendix
\section*{Acknowledgements}
We would like to thank Alberto Sonnino (Mysten Labs), Kushal Babel (Category Labs) and Alexander Spiegelman (Aptos) for discussions on the project, and helping us understand the landscape of the various deployed blockchains.

\section*{Ethical Considerations}
\paragraph{Stakeholders and Scope.}
This work primarily affects blockchain protocol designers, system architects, and researchers developing execution engines and transaction fee mechanisms for parallel execution.
Indirect stakeholders include blockchain users whose transaction fees and execution outcomes depend on these mechanisms, validators or schedulers whose incentives are shaped by fee design, and application developers whose smart contracts interact with parallel execution environments.
As transaction fee mechanisms influence access to blockspace, the results are relevant to auditors, regulators, and ecosystem participants concerned with fairness and manipulation resistance in blockchain infrastructure.

\paragraph{Research Conduct.}
The research is theoretical and analytical, supported by formal modeling and adversarial reasoning. We analyze fee mechanisms under explicitly stated rational-adversary models and do not involve experiments on live systems, user data, or deployed protocols. Prior work is cited carefully, and our contributions are framed as identifying fundamental limitations and design trade-offs rather than prescribing immediate deployment recommendations.

\paragraph{Benefits and Public Interest.}
By identifying impossibility results and previously unrecognized attack vectors in parallel transaction fee mechanisms, this work helps prevent unsafe or economically unsound protocols from being deployed.
Clarifying the trade-offs between user risk, scheduler risk, and shill resistance contributes to more transparent and accountable fee-market design, supporting fairer transaction inclusion and more predictable user costs.
These insights serve the public interest by strengthening the robustness of blockchain systems that increasingly underpin financial and digital infrastructure.

\paragraph{Risks and Limitations.}
The techniques analyzed in this paper, such as shill attacks exploiting parallelism and execution uncertainty, could potentially be misused by adversarial actors to extract value from poorly designed systems.
Profitable shills can waste scarce execution capacity, distort prices, and change inclusion incentives.
Because the paper analyzes mechanism-level issues in \cite{acilan2025transactionfeemarketdesign}'s proposed designs and not vendor-specific implementation vulnerabilities, vendor disclosure is not directly applicable.
Our analysis assumes rational, economically motivated adversaries and does not model fully malicious behavior, denial-of-service attacks, or off-chain collusion. We explicitly document these assumptions, and caution that deploying fee mechanisms outside their modeled conditions may lead to unintended consequences.

\section*{Open Science}

The query below can be run as is on Dune analytics~\cite{DuneAnalytics} to get the current priority to base ratio. 
\begin{lstlisting}[language=SQL, caption={Query used to estimate priority-to-base fee ratios over the last 30 days.}]
WITH tx_fees AS (
  SELECT 
    t.block_time,
    t.priority_fee_per_gas,
    b.base_fee_per_gas,
    CAST(t.priority_fee_per_gas AS DOUBLE) /
      NULLIF(CAST(b.base_fee_per_gas AS DOUBLE), 0)
      AS priority_to_base_ratio
  FROM base.transactions t
  INNER JOIN base.blocks b
    ON t.block_number = b.number
    AND t.block_date = b.date
  WHERE t.block_date >= CURRENT_DATE - INTERVAL '30' day
    AND b.base_fee_per_gas > 0
    AND t.priority_fee_per_gas IS NOT NULL
)
SELECT 
  AVG(priority_to_base_ratio) AS avg_priority_to_base_ratio,
  APPROX_PERCENTILE(priority_to_base_ratio, 0.5)
    AS median_priority_to_base_ratio,
  MIN(priority_to_base_ratio) AS min_ratio,
  MAX(priority_to_base_ratio) AS max_ratio
FROM tx_fees;
\end{lstlisting}

\bibliographystyle{plainurl}
\bibliography{refs}

\begin{thebibliography}{10}

\bibitem{blockbuilding}
Mikerah A. and Sarisht Wadhwa.
\newblock Block building is not just knapsack!, June 2024.
\newblock Ethereum Research forum post.
\newblock URL: \url{https://ethresear.ch/t/block-building-is-not-just-knapsack/19871}.

\bibitem{acilan2025transactionfeemarketdesign}
Bahar Acilan, Andrei Constantinescu, Lioba Heimbach, and Roger Wattenhofer.
\newblock Transaction fee market design for parallel execution, 2025.
\newblock URL: \url{https://arxiv.org/abs/2502.11964}, \href {https://arxiv.org/abs/2502.11964} {\path{arXiv:2502.11964}}.

\bibitem{aptos-ordering}
aptos foundation.
\newblock aip-27 · aptos-foundation/aips, 2022.
\newblock URL: \url{https://github.com/aptos-foundation/AIPs/blob/main/aips/aip-27.md}.

\bibitem{babel2025mysticeti}
Kushal Babel, Andrey Chursin, George Danezis, Anastasios Kichidis, Lefteris Kokoris{-}Kogias, Arun Koshy, Alberto Sonnino, and Mingwei Tian.
\newblock Mysticeti: Reaching the latency limits with uncertified dags.
\newblock In {\em 32nd Annual Network and Distributed System Security Symposium, {NDSS} 2025, San Diego, California, USA, February 24-28, 2025}. The Internet Society, 2025.
\newblock URL: \url{https://www.ndss-symposium.org/ndss-paper/mysticeti-reaching-the-latency-limits-with-uncertified-dags/}.

\bibitem{bahrani2024transaction}
Maryam Bahrani, Pranav Garimidi, and Tim Roughgarden.
\newblock Transaction fee mechanism design with active block producers, 2023.
\newblock URL: \url{https://arxiv.org/abs/2307.01686}, \href {https://arxiv.org/abs/2307.01686} {\path{arXiv:2307.01686}}.

\bibitem{amm}
Henry Berg and Todd~A. Proebsting.
\newblock Hanson's automated market maker.
\newblock {\em Journal of Prediction Markets}, 3(1):45--59, April 2009.
\newblock URL: \url{https://ideas.repec.org/a/buc/jpredm/v3y2009i1p45-59.html}.

\bibitem{chung2023foundations}
Hao Chung and Elaine Shi.
\newblock Foundations of transaction fee mechanism design, 2022.
\newblock URL: \url{https://arxiv.org/abs/2111.03151}, \href {https://arxiv.org/abs/2111.03151} {\path{arXiv:2111.03151}}.

\bibitem{diamandis2023designing}
Theo Diamandis, Alex Evans, Tarun Chitra, and Guillermo Angeris.
\newblock {Designing Multidimensional Blockchain Fee Markets}.
\newblock In {\em 5th Conference on Advances in Financial Technologies (AFT 2023)}, volume 282 of {\em Leibniz International Proceedings in Informatics (LIPIcs)}, pages 4:1--4:23. Schloss Dagstuhl -- Leibniz-Zentrum f{\"u}r Informatik, 2023.
\newblock \href {https://doi.org/10.4230/LIPIcs.AFT.2023.4} {\path{doi:10.4230/LIPIcs.AFT.2023.4}}.

\bibitem{DuneAnalytics}
Dune analytics.
\newblock URL: \url{https://dune.com/home}.

\bibitem{solanafoundation_2019}
Solana Foundation.
\newblock Sealevel - parallel processing thousands of smart contracts, Sep 2019.
\newblock URL: \url{https://solana.com/news/sealevel---parallel-processing-thousands-of-smart-contracts}.

\bibitem{suigas}
SUI Foundation.
\newblock Gas in sui | sui documentation, 2025.
\newblock URL: \url{https://docs.sui.io/concepts/tokenomics/gas-in-sui}.

\bibitem{gafni2026deterring}
Yotam Gafni.
\newblock Deterring a small collusion is all you need.
\newblock In {\em Proceedings of the ACM Web Conference 2026}, WWW '26, pages 31--39, New York, NY, USA, April 2026. Association for Computing Machinery.
\newblock \href {https://doi.org/10.1145/3774904.3792121} {\path{doi:10.1145/3774904.3792121}}.

\bibitem{gafni2022greedy}
Yotam Gafni and Aviv Yaish.
\newblock Greedy transaction fee mechanisms for (non-)myopic miners, 2022.
\newblock URL: \url{https://arxiv.org/abs/2210.07793v4}.

\bibitem{gafni2024barriers}
Yotam Gafni and Aviv Yaish.
\newblock Barriers to collusion-resistant transaction fee mechanisms.
\newblock In {\em Proceedings of the 25th ACM Conference on Economics and Computation}, EC '24, page 1074–1096, New York, NY, USA, 2024. Association for Computing Machinery.
\newblock \href {https://doi.org/10.1145/3670865.3673469} {\path{doi:10.1145/3670865.3673469}}.

\bibitem{gafni2024discrete}
Yotam Gafni and Aviv Yaish.
\newblock Discrete and bayesian transaction fee mechanisms.
\newblock In Stefanos Leonardos, Elise Alfieri, William~J. Knottenbelt, and Panos Pardalos, editors, {\em Mathematical Research for Blockchain Economy}, pages 145--171, Cham, 2024. Springer Nature Switzerland.

\bibitem{gafni2024scheduling}
Yotam Gafni and Aviv Yaish.
\newblock Scheduling with time discounts, 2024.
\newblock \href {https://doi.org/10.48550/arXiv.2402.08549} {\path{doi:10.48550/arXiv.2402.08549}}.

\bibitem{gafni2026transaction}
Yotam Gafni and Aviv Yaish.
\newblock Transaction fee mechanisms robust to welfare-increasing collusion.
\newblock {\em Games and Economic Behavior}, 157:351--375, March 2026.
\newblock \href {https://doi.org/10.1016/j.geb.2026.02.005} {\path{doi:10.1016/j.geb.2026.02.005}}.

\bibitem{garimidi2025transaction}
Pranav Garimidi, Lioba Heimbach, and Tim Roughgarden.
\newblock Transaction fee mechanism design for leaderless blockchain protocols, 2025.
\newblock URL: \url{https://arxiv.org/abs/2505.17885}, \href {https://arxiv.org/abs/2505.17885} {\path{arXiv:2505.17885}}.

\bibitem{greenlaw1995limits}
Raymond Greenlaw, H~James Hoover, and Walter~L Ruzzo.
\newblock {\em Limits to parallel computation: P-completeness theory}.
\newblock Oxford University Press on Demand, Oxford, UK, 1995.
\newblock \href {https://doi.org/10.1093/oso/9780195085914.001.0001} {\path{doi:10.1093/oso/9780195085914.001.0001}}.

\bibitem{heimbach2024defi}
Lioba Heimbach, Quentin Kniep, Yann Vonlanthen, and Roger Wattenhofer.
\newblock Defi and nfts hinder blockchain scalability.
\newblock In Foteini Baldimtsi and Christian Cachin, editors, {\em Financial Cryptography and Data Security}, pages 291--309, Cham, 2024. Springer Nature Switzerland.

\bibitem{heimbach2025dissecting}
Lioba Heimbach, Quentin Kniep, Yann Vonlanthen, Roger Wattenhofer, and Patrick Z{\"u}st.
\newblock Dissecting the~eip-2930 optional access lists.
\newblock In Jeremy Clark and Elaine Shi, editors, {\em Financial Cryptography and Data Security}, pages 292--302, Cham, 2025. Springer Nature Switzerland.
\newblock \href {https://doi.org/10.1007/978-3-031-78676-1_16} {\path{doi:10.1007/978-3-031-78676-1_16}}.

\bibitem{kiayias2025one}
Aggelos Kiayias, Elias Koutsoupias, Giorgos Panagiotakos, and Kyriaki Zioga.
\newblock One-dimensional vs. {Multi}-dimensional {Pricing} in {Blockchain} {Protocols}, June 2025.
\newblock arXiv:2506.13271 [cs].
\newblock \href {https://doi.org/10.48550/arXiv.2506.13271} {\path{doi:10.48550/arXiv.2506.13271}}.

\bibitem{lavee2025does}
Nir Lavee, Noam Nisan, Mallesh Pai, and Max Resnick.
\newblock Does {Your} {Blockchain} {Need} {Multidimensional} {Transaction} {Fees}?, April 2025.
\newblock URL: \url{http://arxiv.org/abs/2504.15438}, \href {https://doi.org/10.48550/arXiv.2504.15438} {\path{doi:10.48550/arXiv.2504.15438}}.

\bibitem{mohan2024blockknapsack}
Vijay Mohan and Peyman Khezr.
\newblock Blockchains, mev and the knapsack problem: A primer, 2024.
\newblock URL: \url{https://arxiv.org/abs/2403.19077}, \href {https://arxiv.org/abs/2403.19077} {\path{arXiv:2403.19077}}.

\bibitem{MonadGas}
Monad.
\newblock Gas on monad | monad developer documentation.
\newblock URL: \url{https://docs.monad.xyz/developer-essentials/gas-on-monad#gas-price-detail}.

\bibitem{Monad}
Monad.
\newblock Parallel execution | monad developer documentation.
\newblock URL: \url{https://docs.monad.xyz/monad-arch/execution/parallel-execution}.

\bibitem{ponnapalli2021rainblock}
Soujanya Ponnapalli, Aashaka Shah, Souvik Banerjee, Dahlia Malkhi, Amy Tai, Vijay Chidambaram, and Michael Wei.
\newblock {RainBlock}: Faster transaction processing in public blockchains.
\newblock In {\em 2021 USENIX Annual Technical Conference (USENIX ATC 21)}, pages 333--347. USENIX Association, July 2021.
\newblock URL: \url{https://www.usenix.org/conference/atc21/presentation/ponnapalli}.

\bibitem{eip1559}
Ethereum~Improvement Proposals.
\newblock Eip-1559: Fee market change for eth 1.0 chain, April 2019.
\newblock URL: \url{https://eips.ethereum.org/EIPS/eip-1559}.

\bibitem{roughgarden2024transaction}
Tim Roughgarden.
\newblock Transaction fee mechanism design.
\newblock {\em J. ACM}, 71(4), August 2024.
\newblock \href {https://doi.org/10/g9z6nw} {\path{doi:10/g9z6nw}}.

\bibitem{spiegelman2022bullshark}
Alexander Spiegelman, Neil Giridhaman, Alberto Sonnino, and Lefteris Kokoris-Kogias.
\newblock Bullshark: Dag bft protocols made practical.
\newblock In {\em Proceedings of the 2022 ACM SIGSAC Conference on Computer and Communications Security}, 2022.

\bibitem{sui2025object}
Sui.
\newblock Object {Model} {\textbar} {Sui} {Documentation}, 2025.
\newblock URL: \url{https://docs.sui.io/concepts/object-model}.

\bibitem{wahrstatter2025execution}
Anton Wahrst{\"a}tter.
\newblock Execution dependencies, April 2025.
\newblock URL: \url{https://ethresear.ch/t/execution-dependencies/22150}.

\bibitem{AccessLists}
Toni Wahrstätter, Dankrad Feist, Francesco D`Amato, Jochem Brouwer, Ignacio Hagopian, Felipe Selmo, Rahul, and Stefan, 2025.
\newblock URL: \url{https://eips.ethereum.org/EIPS/eip-7928}.

\bibitem{yaish2024speculative}
Aviv Yaish, Kaihua Qin, Liyi Zhou, Aviv Zohar, and Arthur Gervais.
\newblock Speculative denial-of-service attacks in ethereum.
\newblock In {\em 33rd USENIX Security Symposium (USENIX Security 24)}, USENIXSEC '24, pages 3531--3548, Philadelphia, PA, 8 2024. USENIX Association.
\newblock URL: \url{https://www.usenix.org/conference/usenixsecurity24/presentation/yaish}.

\end{thebibliography}

\section{Convergence Attacks}
While a recent elegant work has shown that multidimensional fee mechanisms may be slower to converge to a steady state~\cite{kiayias2025one}, the proof relies on a very strict assumption.
Intuitively, the assumption boils down to considering very different convergence time distributions for single- and multi- dimensional mechanisms.
In more detail, it is assumed that the difference between the distribution of convergence times of a single dimensional mechanism and of any given single resource of a multidimensional mechanism is upper-bounded by a single parameter.
Thus, given enough resources, one is likely to sample a ``bad'' relative convergence time.

Now, we present a concrete counter-example highlighting the issues with the result of~\cite{kiayias2025one}.
\begin{proposition}
    A sequence of transactions representing stable demand for each individual resource induces divergent and oscillatory behavior for a single-dimensional mechanism, while a multidimensional mechanism remains stable.
\end{proposition}
\begin{proof}
    We begin by fixing some multidimensional fee mechanism.
    The most reasonable single-dimensional ``analog'' for the multidimensional mechanism would set its gas limit to equal the lowest limit of any single resource of the multidimensional mechanism.
    To see why, note that a lower limit in the single-dimensional case does not make full use of the block-time, and higher limits result in easy-to-execute attacks.
    
    Next, consider a sequence of transactions such that for each block, the mempool receives a number of transactions which equals the number of resources.
    Each of these transactions consumes exactly a single resource, where the amount that is consumed equals exactly the ``target'' amount of the resource.
    In the multidimensional world, this means that the prices would remain stable.
    The multidimensional case is thus trivially faster to converge.
\end{proof}

Moreover, it is unclear if convergence time is an appropriate benchmark.
In fact, we devise a sequence of transactions where the ``bad'' outcome is for a mechanism to remain stable, as that indicates ignorance of the market's demand dynamics.
Then, while a multidimensional mechanism would constantly be in flux, this is the \emph{desired} behavior, stemming from a reaction to fine-grained changes in demand which the single-dimensional mechanism is ignorant of.

\begin{example}
    Consider two resources with a target of $\frac{G}{2}$ and a limit of $G$, and 3 transactions.
    For $i\in\{1,2\}$, let $tx_i$ consume $\frac{G}{4} - \frac{\epsilon}{2}$ of $r_i$, and let $tx_3$ consume $\frac{\epsilon}{2}$ of each of $r_1, r_2$ (i.e., a total of $\frac{\epsilon}{2}$ is consumed).
    Moreover, $tx_1, tx_2$ find the two resources to be perfect substitutes, while $tx_3$ does not (e.g., consider two AMMs for the same token pair, an arbitrageur needs both, but simple traders with a large enough price-impact tolerance may not really have a preference).
    A single-dimensional mechanism would see a constant demand that equals the target ($\frac{G}{2}$), and so would remain stable.
    However, the multi-dimensional mechanism would understand that resources are under-utilized, and thus would lower fees.
\end{example}

At a high level, understanding how markets react to multidimensional fee mechanisms and the corresponding impact on convergence times highlight an interesting question for future work to answer empirically.
Ethereum's recent adoption of EIP-4844 could prove useful to such an endeavor;
In particular, it seems like the variance in transaction fees has become lower since March 13, 2024, when the Dencun upgrade activated EIP-4844 on Ethereum's mainnet.

\section{Priority Schedules}
\label{sec:priorityschedules}
We now analyze currently deployed priority-based TFM schedules under parallel execution.
We show that these mechanisms exploit parallelism sub-optimally.
That is because in the worst case, their \emph{weighted throughput} (total fee collected) is the same as a linear schedule.
We focus on two classes, \greedy (fee-ordered partial orders over conflicts) and \random (randomized ordering), and prove that in the worst case the weighted throughput of each is bounded by the number of threads available for parallel execution.

Transaction fees play a central role in determining the order and timing of transaction execution.  
The notion of schedulers that prioritize transactions based on some measure of ``importance'' is not new \cite{gafni2022greedy,gafni2024scheduling}.
Indeed, every major blockchain platform employs some variant of a \emph{priority scheduler}, where transactions offering higher fees (or utility) are executed before those with lower ones.  
This mechanism reflects the rational principle that users willing to pay more should receive faster and more reliable inclusion.

Block-building based entirely on fees has often been compared to the well-studied \emph{knapsack problem}, where the goal is to maximize total fee revenue under a limited computational or gas budget~\cite{blockbuilding,mohan2024blockknapsack}.
Historically, however, the primary bottleneck in blockchain systems was \emph{consensus latency} rather than \emph{execution throughput}.  
Consequently, a greedy block-building heuristic (selecting transactions in decreasing order of fee per gas and executing them sequentially) became a natural choice, as in Bitcoin and Ethereum.
We refer to such linear, fee-ordered scheduling as the \emph{linear design}. 

However, for modern blockchains like Solana, SUI, Aptos, Monad, the execution has become the bottleneck. Partial orders allow non-conflicting transactions to run concurrently while respecting fee-ordered precedence among conflicts. A representative instance is Sui's fee-ordered partial order over shared objects~\cite{suigas}, built with a fast consensus (Mysticeti~\cite{babel2025mysticeti}). Transactions declare object usage (via Move Language~\cite{sui2025object}), enabling a fee-ordered DAG that preserves output equivalence while attempting parallel execution.
This \greedy policy seems rational as higher fee implies higher priority \cite{gafni2022greedy}, but, can be revenue-suboptimal under congestion because early high-fee conflicts can block parallel opportunities elsewhere.

\begin{example}
\label{ex:suisuboptimal}
Consider $\{\tx_1,\tx_2,\tx_3,\tx_4\}$ with
$\tx_1 \simeq (200,4,\{\object_1,\object_3\})$,
$\tx_2 \simeq (150,3,\{\object_1,\object_2\})$,
$\tx_3 \simeq (100,2,\{\object_2\})$,
$\tx_4 \simeq (100,1,\{\object_3\})$,
all writes, and block computation limit $\gaslimit=400$.
Under \greedy (fee-ordered DAG), precedence edges $\tx_1\!\to\!\tx_2$, $\tx_2\!\to\!\tx_3$, and $\tx_1\!\to\!\tx_4$ yield
$
  \tx_1 \rightarrow (\tx_2,\tx_4) \rightarrow \tx_3,
$
which exceeds $\gaslimit$, dropping $\tx_3$.
Yet, an alternative valid schedule
$
  (\tx_1,\tx_3) \rightarrow (\tx_2,\tx_4)
$
respects $\gaslimit$ and admits all four, strictly improving revenue.
\cref{fig:sui-suboptimal} illustrates the three schedules (linear, \greedy, parallel optimum) and their object usage.
\end{example}
The example demonstrates that fee-ordered schedules can waste available parallel compute which could have executed lower priority non-conflicting transactions.

\begin{figure*}
\centering
\begin{subfigure}[b]{0.28\linewidth}
\centering
\begin{tikzpicture}[
  tx1/.style={draw, rectangle, fill=green!20, minimum width=1.2cm, minimum height=0.5cm, font=\footnotesize},
  tx2/.style={draw, rectangle, fill=yellow!30, minimum width=1.2cm, minimum height=0.5cm, font=\footnotesize},
  tx3/.style={draw, rectangle, fill=blue!20, minimum width=1.2cm, minimum height=0.5cm, font=\footnotesize},
  tx4/.style={draw, rectangle, fill=red!20, minimum width=1.2cm, minimum height=0.5cm, font=\footnotesize},
  every node/.style={font=\footnotesize},
  ->, >=Latex
]

\node[tx1] (tx1) {$\tx_1$};
\node[tx2, below =0.5cm of tx1] (tx2) {$\tx_2$};
\node[tx3, below =0.5cm of tx2] (tx3) {$\tx_3$};
\node[tx4, below =0.5cm of tx3] (tx4) {$\tx_4$};

\draw (tx1) -- (tx2);
\draw (tx2) -- (tx3);
\draw (tx3) -- (tx4);

\node[below=0.4cm of tx4] {Linear Schedule ($\tx_3, \tx_4$ dropped)};
\end{tikzpicture}
\caption{Linear fee-based order}

\vspace{0.5em}
\begin{tabular}{|>{\centering\arraybackslash}m{1.2cm}|
                >{\centering\arraybackslash}m{1.2cm}|
                >{\centering\arraybackslash}m{1.2cm}|}
\hline
\rowcolor{gray!30}
$\object_1$ & $\object_2$ & $\object_3$ \\
\hline
\cellcolor{green!20} &  & \cellcolor{green!20} \\
\cellcolor{yellow!30} & \cellcolor{yellow!30} &  \\
\hline
\end{tabular}
\end{subfigure}
\hfill
\begin{subfigure}[b]{0.28\linewidth}
\centering
\begin{tikzpicture}[
  tx1/.style={draw, rectangle, fill=green!20, minimum width=1.2cm, minimum height=0.5cm, font=\footnotesize},
  tx2/.style={draw, rectangle, fill=yellow!30, minimum width=1.2cm, minimum height=0.5cm, font=\footnotesize},
  tx3/.style={draw, rectangle, fill=blue!20, minimum width=1.2cm, minimum height=0.5cm, font=\footnotesize},
  tx4/.style={draw, rectangle, fill=red!20, minimum width=1.2cm, minimum height=0.5cm, font=\footnotesize},
  every node/.style={font=\footnotesize},
  ->, >=Latex
]

\node[tx1] (tx1) {$\tx_1$};
\node[tx2, below left=0.4cm and 0.3cm of tx1] (tx2) {$\tx_2$};
\node[tx4, below right=0.4cm and 0.3cm of tx1] (tx4) {$\tx_4$};
\node[tx3, below=0.5cm of tx2] (tx3) {$\tx_3$};

\draw (tx1) -- (tx2);
\draw (tx1) -- (tx4);
\draw (tx2) -- (tx3);

\node[below=3cm of tx1] {Fee-Based DAG Schedule ($\tx_3$ dropped)};
\end{tikzpicture}
\caption{\greedy}

\vspace{0.5em}
\begin{tabular}{|>{\centering\arraybackslash}m{1.2cm}|
                >{\centering\arraybackslash}m{1.2cm}|
                >{\centering\arraybackslash}m{1.2cm}|}
\hline
\rowcolor{gray!30}
$\object_1$ & $\object_2$ & $\object_3$ \\
\hline
\cellcolor{green!20} &  & \cellcolor{green!20} \\
\cellcolor{yellow!30} & \cellcolor{yellow!30} &  \cellcolor{red!20}\\
\hline
\end{tabular}
\end{subfigure}
\hfill
\begin{subfigure}[b]{0.28\linewidth}
\centering
\begin{tikzpicture}[
  tx1/.style={draw, rectangle, fill=green!20, minimum width=1.2cm, minimum height=0.8cm, font=\footnotesize},
  tx2/.style={draw, rectangle, fill=yellow!30, minimum width=1.2cm, minimum height=0.8cm, font=\footnotesize},
  tx3/.style={draw, rectangle, fill=blue!20, minimum width=1.2cm, minimum height=0.8cm, font=\footnotesize},
  tx4/.style={draw, rectangle, fill=red!20, minimum width=1.2cm, minimum height=0.8cm, font=\footnotesize},
  every node/.style={font=\footnotesize},
  ->, >=Latex
]

\node[tx1] (tx1) {$\tx_1$};
\node[tx3, right=.5cm of tx1] (tx3) {$\tx_3$};
\node[tx2, below=.5cm of tx1] (tx2) {$\tx_2$};
\node[tx4, right=.5cm of tx2] (tx4) {$\tx_4$};

\node[below left=0.4cm and -2cm of tx4] {Alternative parallel schedule};
\node[below left=0.9cm and -1.5cm of tx4]
{(Without ordering)};
\end{tikzpicture}
\caption{Optimal parallel order}

\vspace{0.5em}
\begin{tabular}{|>{\centering\arraybackslash}m{1.2cm}|
                >{\centering\arraybackslash}m{1.2cm}|
                >{\centering\arraybackslash}m{1.2cm}|}
\hline
\rowcolor{gray!30}
$\object_1$ & $\object_2$ & $\object_3$ \\
\hline
\cellcolor{green!20} & \cellcolor{blue!20} & \cellcolor{green!20} \\
\cellcolor{yellow!30} & \cellcolor{yellow!30} & \cellcolor{red!20} \\
\hline
\end{tabular}
\end{subfigure}

\caption{Comparison of scheduling strategies and their object usage.}
\label{fig:sui-suboptimal}
\end{figure*}
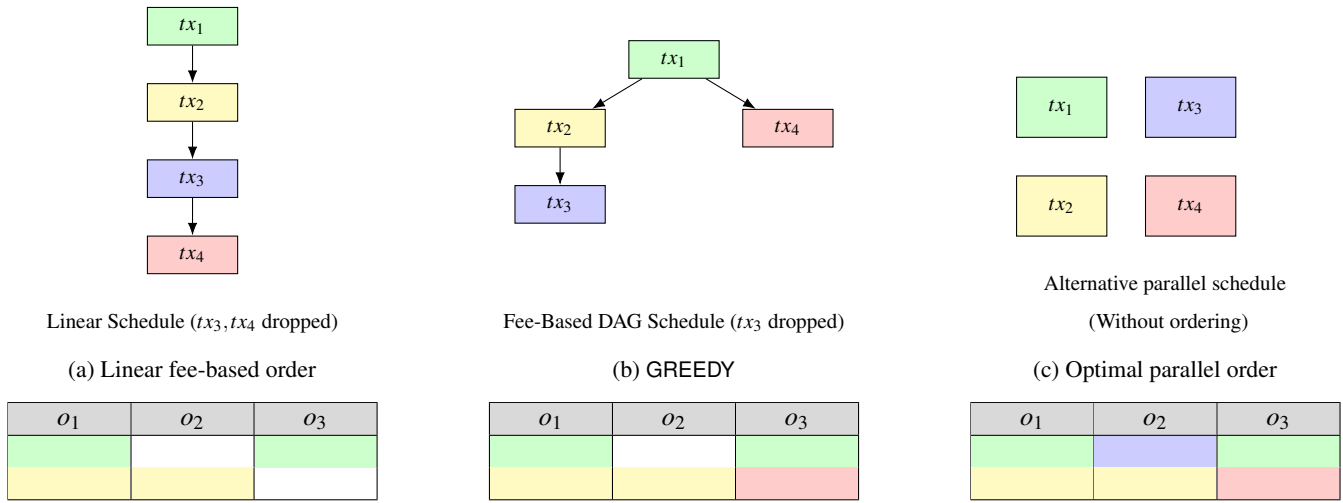

For each schedule $\priorityschedules$, define $\ratio(\priorityschedules)$ as above. With an $\varepsilon\!\sim\!0$ transaction priced at $\sim 1+\varepsilon$ and a transaction with compute $\gaslimit$ priced at $1$, the following lower bound demonstrates that \greedy can be arbitrarily poor:
\[
  \ratio(\greedy)=\frac{\varepsilon\cdot(1+\varepsilon)}{\gaslimit}\approx 0
  .
\]
Fortunately, in practice, systems cap transaction size at $\txcomplimit$, leading to the following worst-case guarantee.

\begin{lemma}
\label{lem:greedylowerbound}
Given $\numcores$ parallel cores and $\numobjects$ objects,
\[
  \ratio(\greedy) \;=\; \frac{\gaslimit-\txcomplimit}{\gaslimit\cdot\min(\numcores,\numobjects)},
\]
and the lower bound is achievable in the limit.
\end{lemma}

\begin{proof}

    Consider a single core machine, and enough transactions to fill the schedule in an optimal schedule. 
    Consider the last scheduled transaction in $\greedy$ as $\tx$. 
    To get relative revenue, consider that this transaction pays $1$ unit per computation used. 
    Thus, any transaction scheduled must pay a gas price $\geq 1$ and any transaction that fails to get scheduled pay a gas price $\leq 1$.
    If the computation used by the schedule is $\leq \gaslimit-\txcomplimit$, then another transaction can be scheduled after $\greedy$, which violates the definition of $\greedy$. 
    Thus, the computation used by the schedule is $> \gaslimit-\txcomplimit$, and the gas price is $\geq 1$.
    This implies that the revenue received is $> \gaslimit-\txcomplimit$.
    Adding multiple cores cannot lower the revenue received, so the revenue generated is $\revenue\greedy > \gaslimit-\txcomplimit$.

    The best possible revenue from \opt would be when the complete block is scheduled with as many transactions in parallel as possible. 
    This would include the revenue of the greedy schedule, and all parallel transactions without scheduling $\tx$. Thus, $\revenue\opt \leq \revenue\greedy+ \txcomplimit\cdot1 + (\gaslimit\cdot 1)\cdot(\text{number of parallel transactions - 1}) < \gaslimit \cdot\text{number of parallel transactions}$. The number of parallel transactions can be at most the number of cores, and also at most the number of objects (otherwise, by pigeon hole principle, two transactions in parallel would have to use the same object).
    Thus, $\text{number of parallel transactions} = \min(\numcores,\numobjects)$, and, \[\alpha(\greedy) = \frac{\revenue\greedy}{\revenue\opt} > \frac{\gaslimit-\txcomplimit}{\gaslimit \cdot \min(\numcores,\numobjects)}.\]

    Next, we show that this revenue lower bound is achievable.
    Consider a set of transactions $\consensusset_s = \{ \tx_i\}$ such that $\tx_i \simeq (\compunits_i, \{\object_1\})$, gas price is $1+\varepsilon$ for each of them and $\sum_i \compunits_i = \gaslimit-\txcomplimit$.
    Next, consider a transaction $\tx \simeq (\varepsilon, \setobjects)$, gas price is $1$. 
    Finally, consider another set of transactions  $\consensusset_n = \{\tx_j\}$ such that $\forall \object_j \in \setobjects: \tx_j \simeq (\txcomplimit, \{\object_j\})$ with a gas price of $1-\varepsilon$ and consider enough instances of each $\tx_j$ to fill the computation limits of the block. 
    If the partial order follows $\greedy$, then only transactions in $\consensusset_s$ and $\tx$ would be scheduled, leading to a revenue of:
    \[\revenue\greedy(\consensusset_s,\tx,\consensusset_n) = \overbrace{(1+\varepsilon)\cdot(\gaslimit-\txcomplimit)}^{\consensusset_s} + \overbrace{1\cdot\varepsilon}^{\tx}\]
    Considering \txcomplimit to be a factor of \gaslimit, the optimal revenue in this case would be:
    \begin{align*}
     \revenue\opt(\consensusset_s,\tx, &\consensusset_n) = 
      \overbrace{(1+\varepsilon)\cdot(\gaslimit-\txcomplimit)}^{\consensusset_s} + \\& \overbrace{(1+\varepsilon)\cdot\txcomplimit}^{\tx_1 \in \consensusset_n} + 
      \overbrace{(1-\varepsilon)\cdot \gaslimit \cdot (\min(\numcores, \numobjects)-1)}^{\forall \tx_j \in \consensusset_n, j\neq 1}
    \end{align*}
    Thus,
    $\lim_{\varepsilon \rightarrow 0^+} \ratio(\greedy) = \frac{\gaslimit-\txcomplimit}{\gaslimit\cdot (\min(\numcores, \numobjects)}.$
\end{proof}

\paragraph{Randomized Order.}
Many systems implement \random: e.g., Aptos performs a random shuffle (with anti-same-sender heuristics) after consensus~\cite{aptos-ordering}; in Monad~\cite{Monad}, consensus fixes order, however, no ordering policy is specified to the consensus.
This would yield a distribution of transactions that we treat as random for analysis.
Since \greedy is a valid instantiation of a random ordering, the worst case for \random cannot be better:

\begin{lemma}
The lower bound for $\ratio(\random)$ is at least as bad as the lower bound for $\ratio(\greedy)$.
\end{lemma}
\begin{proof}
Any \greedy output is feasible for \random, so \random's worst case is no better than \greedy's.
\end{proof}

\paragraph{Remark on Optimistic Execution (Monad).}
In optimistic execution, transactions start immediately after consensus without a precomputed schedule; conflicts trigger re-execution. If a fee-ordered consensus sequence places a large set $\consensusset_s$ of mutually conflicting transactions first, many cores can be kept busy executing work that later aborts, deferring useful parallelism. When early transactions in $\consensusset_s$ use $\varepsilon$ compute while others consume near $\txcomplimit$, the first successful commit delays re-execution by $\approx\txcomplimit$, potentially yielding revenue strictly below a simple \random baseline despite high core utilization (wasted on aborts). This phenomenon does not alter the proofs above; it illustrates that \emph{execution policy} can further degrade realized revenue beyond ordering alone.

\end{document}